%% file: hgraphs.tex
\documentclass[12pt]{article}
\usepackage[utf8]{inputenc}
\usepackage{amsthm,amssymb,amsmath, amsfonts}
\usepackage{graphicx,algorithmic,algorithm}
\usepackage[]{datetime}
\usepackage[nottoc,numbib]{tocbibind}
\usepackage{cite}
\usepackage{fullpage}
\usepackage{paralist}
\usepackage{multirow}
\usepackage{xspace}
\usepackage{tikz}
\usepackage{hyperref}
\usetikzlibrary{decorations.pathreplacing}
\tikzset{black node/.style={draw, circle, fill = black, minimum size = 5pt, inner sep = 0pt}}
\tikzset{white node/.style={draw, circle, fill = white, minimum size = 5pt, inner sep = 0pt}}
\tikzset{normal/.style = {draw=none, fill = none}}

\hypersetup{
  pdftitle = {On the tractability of optimization problems on H-graphs},
  pdfauthor=  {F. Fomin, P. Golovach, J.-F. Raymond},
  colorlinks = true
}

\theoremstyle{plain}
  \newtheorem{theorem}{Theorem}
  
  \newtheorem{corollary}{Corollary}
  \newtheorem{lemma}{Lemma}
   \newtheorem{claimn}{Claim}
    
\theoremstyle{definition}
  \newtheorem{definition}{Definition}
\theoremstyle{remark}
  \newtheorem{remark}{Remark}

   \def\cqed{\ifmmode$\lrcorner$\else{\unskip\nobreak\hfil
         \penalty50\hskip1em\null\nobreak\hfil$\lrcorner$
         \parfillskip=0pt\finalhyphendemerits=0\endgraf}\fi}

  \newcommand{\N}{\mathbb{N}}

  \DeclareMathOperator{\dist}{dist}
  
  \DeclareMathOperator{\nec}{nec}
  
  \DeclareMathOperator{\mim}{mim}
  \DeclareMathOperator{\boolw}{boolw}

\DeclareMathOperator{\operatorClassFPT}{FPT\xspace}
\newcommand{\classFPT}{\ensuremath{\operatorClassFPT}\xspace}
\DeclareMathOperator{\operatorClassW}{W}
\newcommand{\classW}[1]{\ensuremath{\operatorClassW[#1]}}

\DeclareMathOperator{\Oh}{\mathcal{O}}

\begin{document}

\title{On the tractability of optimization problems on $H$-graphs\thanks{An extended abstract of this work appeared in the proceedings of ESA 2018 \cite{fomin_et_al:LIPIcs:2018:9493}. The first two authors have been supported by the Research Council of Norway via the projects ``CLASSIS'' and ``MULTIVAL''. The third author has been supported by the Polish National Science Centre grant PRELUDIUM DEC-2013/11/N/ST6/02706 and by the European Research Council (ERC) under the European
    Union's Horizon 2020 research and innovation programme, ERC consolidator grant DISTRUCT, agreement No 648527.}}

\author{Fedor V.\ Fomin\thanks{
    Department of Informatics, University of Bergen, Norway.} \addtocounter{footnote}{-1}
  \and
  Petr A.\ Golovach\footnotemark{}
  \and
  Jean-Florent Raymond\thanks{CNRS, LIMOS, Université Clermont Auvergne, France.}}
\date{}

\maketitle

 \begin{abstract}
For a graph $H$, a graph $G$ is an $H$-graph if it is an intersection graph of connected subgraphs of some subdivision of $H$.  $H$-graphs naturally generalize several important graph classes like interval graphs or circular-arc graph. This class was introduced in the early 1990s by B\'\i r\'o,  Hujter, and Tuza. Recently, Chaplick et al. initiated the algorithmic study of $H$-graphs by showing that a number of fundamental optimization problems like \textsc{Maximum Clique}, \textsc{Maximum Independent Set}, or \textsc{Minimum Dominating Set} are solvable in polynomial time on $H$-graphs. We extend and complement these algorithmic findings in several directions. 
  
First we show that for every fixed $H$, the class of $H$-graphs is of logarithmically-bounded  boolean-width (via mim-width). Pipelined with the plethora of known algorithms on graphs of bounded boolean-width, this describes a large class of problems solvable in polynomial time on $H$-graphs. We also observe that $H$-graphs are graphs with polynomially many minimal separators. Combined with the work of Fomin, Todinca and Villanger on algorithmic properties of such classes of graphs, this identify another wide class of problems solvable in polynomial time on $H$-graphs.

The most fundamental optimization problems among the problems solvable in polynomial time on $H$-graphs are \textsc{Maximum Clique}, \textsc{Maximum Independent Set}, and \textsc{Minimum Dominating Set}. We provide a more refined complexity analysis of these problems from the perspective of parameterized complexity. We show that \textsc{Maximum Independent Set} and \textsc{Minimum Dominating Set} are W[1]-hard being parameterized by the size of $H$ plus the size of the solution. On the other hand, we prove that when $H$ is a tree, then \textsc{Minimum Dominating Set} is fixed-parameter tractable (FPT) parameterized by the size of~$H$.

For \textsc{Maximum Clique} we show that it admits a polynomial kernel parameterized by $H$ and the solution size.   
\end{abstract}

\section{Introduction}\label{sec:intro}

The notion of $H$-graph was introduced in the work of B\'\i r\'o,  Hujter,  and Tuza
\cite{BIRO1992267} on precoloring extensions of graphs.  $H$-graphs nicely  generalize several popular and widely studied classes of graphs. For example, the classical definition of an interval graph is as a  graph which is an intersection graph\footnote{The intersection graph of a family $\mathcal{S}$ of sets has vertex set $\mathcal{S}$ and edge set $\{SS',\ S\cap S' \neq \emptyset\}$.} of intervals of a  line. Equivalently, a graph is interval if it is an intersection graph of some subpaths of a path. Or, equivalently, if it is an intersection graph of some subgraphs of some subdivision (which is a graph obtained by placing vertices of degree $2$ on the edges) of $P_2$, the graph with two adjacent vertices. Similarly, every chordal graph is an intersection graph of subtrees of some tree. More generally, for a fixed graph $H$,  a graph $G$ is an $H$-graph if it is an intersection graph of some connected subgraphs of some subdivision of $H$. Thus for example, an interval graph is a $P_2$-graph,  a circular-arc graph is a $C_2$-graph, where $C_2$ is a double-edge with two endpoints,  a split graph is a $K_{1,d}$-graph for some $d\geq 0$, where $K_{1,d}$ is a star with $d$ leaves, etc.

The main motivation behind the study of 
$H$-graphs is the following. It is well-known that on interval, chordal, circular-arc, and other graphs with ``simple'' intersection models many NP-hard optimization problems are solvable in polynomial time, see e.g.\ the book of Golumbic \cite{Golumbic04} for an overview.  It is a natural question whether at least some of these algorithmic results can be extended to more general classes of intersection graphs.  
Chaplick et al.
\cite{ChaplickTVZ16} and 
Chaplick and Zeman 
\cite{ChaplickZ17} initiated the systematic study of algorithmic properties of $H$-graphs. They showed that a number of fundamental optimization problems like 
\textsc{Maximum Independent Set} and \textsc{Minimum Dominating Set} are solvable in polynomial time on $H$-graphs for any fixed $H$. 
 Most of the algorithms developed on $H$-graphs in \cite{ChaplickTVZ16,ChaplickZ17}  run in time $n^{f(H)}$, where $n$ is the number of vertices in the input graph and $f$ is some function. In other words, being parameterized by $H$ most of the problems are known to be in the class XP.

Our work is driven by the following  question.
\begin{itemize}
\item Are there  generic explanations why many problems admit polynomial time algorithms on $H$-graphs?
\end{itemize}

We address the first question by proving the following combinatorial results. 
 We show first that  every $H$-graph has mim-width (a graph parameter to be defined in the corresponding section) at most $2|E(H)|+1$. Moreover,  a decomposition of mim-width $2|E(H)|+1$ can be found in polynomial time. Using known inequalities, this gives upper-bounds on the boolean-width of $H$-graphs. This combinatorial result extends the results of Belmonte and Vatshelle \cite{BELMONTE201354, belmonte2010graph} on the boolean-width of interval (resp.\ circular-arc) graphs to $H$-graphs.  Together with the algorithms 
 for a vast class of problems  called \emph{\textsf{LC-VSP}} problems \cite{BUIXUAN20115187,BELMONTE201354} and their distance versions \cite{jaffke2018generalized},
 which are solvable on $n$-vertex graphs of boolean-width $b$ in time $2^{b}\cdot n^{\Oh(1)}$, this implies immediately that all these problems are solvable in polynomial time on $H$-graphs, when $H$ is fixed. The illustrative problems solvable in polynomial time on $H$-graphs by making use of this approach are \textsc{Maximum Weight Independent Set}, \textsc{Minimum Weight Dominating Set}, \textsc{Total Dominating Set}, \textsc{Induced Matching}, and many others. We also obtain polynomial-time algorithms for problems related to induced paths such as \textsc{Longest Induced Path} and \textsc{Disjoint Induced Paths} using the results of Jaffke, Kwon, and Telle \cite{jaffke2017polynomial}. Incidentally, these results demonstrate the applicability of the parameter mim-width.
 
Then we prove that   every $n$-vertex $H$-graph  has at most $(2n + 1)^{|E(H)|} + |E(H)|\cdot
(2n)^2$ minimal separators.\footnote{It was reported to us by Steven Chaplick and Peter Zeman that they also obtained this result independently and that it will be included in the journal version of their paper.} 
Pipelining the bound on the number of minimal separators in $H$-graphs with meta-algorithmic results %about induced subgraphs of bounded treewidth which properties are expressible in Counting Monadic Second Order Logic (CMSOL)
of Fomin, Todinca and Villanger \cite{Fomin:2015:LIS:3065779.3065785}, we obtained another wide class of problems solvable in polynomial time on $H$-graphs. Examples of such problems are
\textsc{Treewidth}, \textsc{Minimum Feedback Vertex Set}, \textsc{Maximum Induced Subgraph excluding a planar minor}, and various packing problems.

All these generic algorithmic results provide XP algorithms when parameterized by the size of $H$. This brings us immediately  to the second question defining the direction of our research.
\begin{itemize}
%\item Is there a generic explanation why many problems admit polynomial time algorithms on $H$-graphs? 
 \item What is the parameterized complexity of the fundamental optimization problems being parameterized by the size of $H$?
\end{itemize}

The first steps in this direction were done by  Chaplick et al. in \cite{ChaplickTVZ16} who showed that \textsc{Minimum Dominating Set} is fixed-parameter tractable (FPT) on $K_{1,d}$-graphs parameterized by $d$. In this paper we show that  \textsc{Minimum Dominating Set} is W[1]-hard parameterized by the size of $H$ plus the solution size. Thus the existence of an FPT algorithm for a general graph $H$ is very unlikely. (We refer to books  \cite{DowneyFbook13,CFKLMPPS2014} for definitions from parameterized complexity and algorithms.)
We also prove a similar lower bound for  \textsc{Maximum Independent Set}  parameterized by the size of $H$ plus the solution size. 
Combined with our combinatorial results, these lower-bounds show that \textsc{Maximum Independent Set} and \textsc{Minimum Dominating Set} are also \classW{1}-hard when parameterized by mim-width of the input and the solution size. The technique we develop to establish lower bounds on $H$-graphs  found applications beyond the topic of this paper 
\cite{jaffke17noteKT,jaffke2017polynomial}. 

On the positive side, we show that when $H$ is a tree, then  \textsc{Minimum Dominating Set} is FPT parameterized by the size of~$H$. This significantly extends the result from  
 \cite{ChaplickTVZ16} for stars to arbitrary trees. Furthermore, our algorithm does not require the intersection representation of the input graph to be given. We actually prove a slightly more general result, namely that  \textsc{Minimum Dominating Set} is FPT on 
   chordal graphs $G$ parameterized by the leafage of the graph, i.e.\ the minimum number of leaves in a clique tree of $G$.
   
   Finally we show that \textsc{Clique} admits a polynomial kernel when parameterized by the size of $H$ plus the solution size. This strengthens the result of Chaplick and Zeman who showed that \textsc{Clique} is FPT for such a parameterization.
   Our algorithmic results about $H$-graphs are summarized in \autoref{tab:sum}.

      \begin{table}[h]
  \centering
\begin{tabular}[h]{|llllll|}
  \hline
  Problem & Parameters & Restrictions & Repr. &Complexity & Ref.\\
  \hline
  any \textsf{LC-VSP}-problem & $\|H\|$ & none & Y &XP & \autoref{thm:lc-vsphg}\\
  \hline
  Induced path problems & $\|H\|$ & none & Y & XP &\autoref{indpa}\\
  \hline
  % \parbox{5cm}{\vspace{0.5mm}\textsc{Opt.\ Ind.\ Sgr.\ for $\mathcal{P}$ and $t$}, when $\mathcal{P}$ is CMSOL} & $\|H\|$ & none & XP & \autoref{c:oispt}\\
  $\textsc{OIS}(\mathcal{P},t)$, $\mathcal{P}$ is CMSOL & $\|H\|$ & none & N & XP & \autoref{c:oispt}\\
  \hline
  \multirow{5}{*}{\textsc{Dominating Set}}          & none & $H$ is a tree & N & NP-hard &\cite{doi:10.1137/0211015}\\
  \cline{2-5}
                       &\multirow{3}{*}{$\|H\|$} & $H$ is a star & N & FPT & \cite{ChaplickTVZ16}\\
  \cline{3-5}
          &  & $H$ is a tree & N & FPT & \autoref{thm:ds-leafage}\\
  \cline{3-5}
          &  & none & Y & XP & \cite{ChaplickTVZ16}\\
  \cline{2-5}
          & $\|H\| +k$ & none & Y & W[1]-hard & \autoref{thm:ds-is}\\
  \hline
  \multirow{4}{*}{\textsc{Independent Set}}& none & $H$ is a tree & N & polynomial & \cite{Gavril72}\\
  \cline{2-5}
          & $\|H\|$ & none & Y & XP &  \cite{ChaplickTVZ16}\\
  \cline{2-5}
          & $\|H\|+k$ & none & Y & W[1]-hard & \autoref{thm:ds-is}\\
  \hline
  \multirow{4}{*}{\textsc{Clique}} & none & $H$ is a tree & N & polynomial & \cite{Gavril72}\\
%    \cline{3-5}
%          &   & $H$ is a cactus & N & polynomial & \cite{ChaplickZ17}\\
  \cline{2-5}
          & $\|H\|$  & none & Y & para-NP-hard & \cite{ChaplickZ17}\\
  \cline{2-5}
          & \multirow{2}{*}{$\|H\| +k$} & \multirow{2}{*}{none} & N & FPT & \cite{ChaplickZ17}\\
  \cline{4-5}
          & & & Y & poly.\ kernel & \autoref{thm:kernel}\\
  \hline
\end{tabular}
  
  \caption[Summary of algorithmic results on $H$-graphs]{Summary of algorithmic results on $H$-graphs, including the classic results on chordal graphs ($H$ is a tree). The fourth column indicates whether a representation of the input as an $H$-graph is given. For each of the mentioned problems, $k$ denotes the solution size. See Sections~\ref{sec:bool} and \ref{sec:minsep}, for details about the first three problem.}
  \label{tab:sum}
\end{table}

\paragraph{Organization of the paper.} 
  \autoref{sec:def} contains the necessary definitions. In \autoref{sec:bool}, we upper-bound the boolean-width of $H$-graphs and provide algorithmic applications. \autoref{sec:minsep} is devoted to the study of minimal separators in $H$-graphs, again with algorithmic consequences. Finally, \autoref{sec:pc} contains our results on the parameterized complexity of some classic optimization problems on $H$-graphs.

\section{Definitions}
\label{sec:def}

\paragraph{Basics.}
All graphs in this paper are finite, undirected, loopless, and may
have multiple edges. If $G$ is a graph, we respectively denote by $|G|$ and $\|G\|$ its numbers of vertices and edges (counting
multiplicities).
If $X$ and $Y$ are disjoint subsets of $V(G)$, $\overline{X}$ is the complement of $X$ in $V(G)$ (i.e.\ $\overline{X} = V(G) \setminus X$), $G[X]$ is the subgraph of $G$ induced by the vertices of $X$, and $G[X,Y]$ is the bipartite subgraph of $G$ with vertex set $X \cup Y$ and as edge set those edges of $G$ that have one endpoint in $X$ and the other in~$Y$.
Unless otherwise specified, logarithms are binary.

\paragraph{$H$-graphs.}
Let $H$ be a (multi) graph. We say that a graph $G$ is an $H$-graph if there
is a subdivision $H'$ of $H$ and a collection $\mathcal{M} =
\{M_v\}_{v \in V(G)}$ (called an \emph{$H$-representation} or, simply, \emph{representation})
of subsets of $V(H')$, each inducing a connected subgraph, such that $G$
is isomorphic to the intersection graph of $\mathcal{M}$. 
To avoid confusion, we refer to the vertices of $H'$ as \emph{nodes}.
We also say that the nodes of $H$ are \emph{branching} nodes of $H'$ and the other nodes are \emph{subdivision} nodes.
If $v$ is a vertex of $G$, then $M_v$ is the \emph{model} of $v$ in
the representation~$\mathcal{M}$.

For every set $A \subseteq V(G)$, we define $M_A = \bigcup_{v \in A} M_v$.
For every node $u$ of $H'$, we denote by $V_u$ the set of vertices of
$G$ whose model contains~$u$, that is,
  \[
    V_u = \{v \in V(G),\ u \in M_v\}.
  \]
  
\paragraph{Parameterized Complexity.} We refer to the  books~\cite{DowneyFbook13,CFKLMPPS2014} for the detailed introduction to the field. Here we only briefly review the basic notions.

Parameterized Complexity is a two dimensional framework
for studying the computational complexity of a problem. One dimension is the input size~$|I|$ of an instance $I$ of the problem and the other is a \emph{parameter}~$k$ associated with the input. 
A parameterized problem is said to be \emph{fixed-parameter tractable} (or FPT) if it can be solved in time $f(k)\cdot |I|^{\Oh(1)}$ for some function~$f$.
The parameterized complexity class FPT consists of all fixed-parameter tractable problems. A parameterized problem is in the class XP if it can be solved in time $|I|^{f(k)}$ for a function $f$. 
Note that  if a parameterized problem is NP-hard for some \emph{fixed} value of the parameter, then it is said that the problem is \emph{para-NP-hard} and it cannot be in XP (and, therefore, in FPT) unless ${\rm P}={\rm NP}$.   Parameterized Complexity also provides special tools to refute the FPT  algorithms under plausible complexity-theoretic assumptions. The main assumption  is the conjecture that ${\rm FPT}\neq {\rm W}[1]$ for the parameterized complexity class W[1]  that play a central role in obtaining lower complexity bounds. The basic way to show that it is unlikely that a parameterized problem admit an FPT algorithm is to show that it is W[1]-hard using a \emph{parameterized reduction} from a known W[1]-hard problem. 

A \emph{kernelization} for a parameterized problem is a polynomial algorithm that maps each instance $(I,k)$ of a parameterized problem with the input~$I$ and parameter~$k$ to an instance $(I',k')$ of the same problem such that
\begin{itemize}
\item[(i)] $(I,k)$ is a yes-instance if and only if $(I',k')$ is a  yes-instance, and
\item[(ii)] $|I'|+k'$ is bounded by~$f(k)$ for a computable function~$f$.
\end{itemize}
The output $(I',k')$ is called a \emph{kernel}. The function~$f$ is said to be the \emph{size} of the kernel. A kernel is \emph{polynomial} if~$f$ is polynomial.
While it can be shown that every decidable parameterized problem is { FPT} if and only if it admits a kernel, it is unlikely that every problem in {FPT} has a polynomial kernel up to certain complexity assumptions. We refer to the aforementioned books and to \cite{fomin2019kernelization} for more details.

\section{\texorpdfstring{$H$}{H}-graphs have logarithmic boolean-width}
%\section{On the mim-width and boolean-width of $H$-graphs}
  \label{sec:bool}

 Boolean-width is a graph invariant that has been introduced in \cite{BUIXUAN20115187} and which is related to the number of different neighborhoods along a cut. Belmonte and Vatshelle showed in \cite{BELMONTE201354} that $n$-vertex interval graphs and circular-arc graphs have boolean-width $\Oh(\log n)$. In this section, we generalize their result by proving that, for any fixed graph $H$, $n$-vertex $H$-graphs have boolean-width~$\Oh(\log n)$. This is done by first upper-bounding the mim-width of $H$-graphs by $\max \{1, 2\|H\|\}$ (\autoref{mim2h}). Using the results of \cite{BUIXUAN201366, jaffke2017polynomial, jaffke2018generalized}, we obtain polynomial time algorithms for a vast class of optimization problems on $H$-graphs. Before we proceed with the proofs, we need to introduce some notions specific to this section.

\begin{definition}
A \emph{branch decomposition} of a graph $G$ is a pair $(T, \delta)$
where $T$ is a full binary rooted tree (that is, every non-leaf vertex has degree 3) and $\delta$ is a bijection from the leaves of $T$
to the vertices of~$G$. A branch decomposition $(T, \delta)$ is a \emph{caterpillar decomposition} if $T$ can be obtained from a path by adding a vertex of degree one adjacent to every internal vertex.
If $w\in V(T)$, let us denote by $V_w$ the set of vertices of $G$ in bijection with the leaves of the subtree of $T$ rooted at~$w$.
\end{definition}

\begin{definition}[maximum induced matching along a cut]
  A set of vertices of a graph $G$ is an \emph{induced matching} if it induces a disjoint union of edges. If $X \subseteq V(G)$, we denote by  $\mim_G(X)$ the maximum number of edges in an induced matching of $G[X, \overline{X}]$. We drop the subscript when there is no ambiguity.
  If $(T, \delta)$ is a branch decomposition of $G$, we denote by $\mim(T, \delta)$ the maximum of $\mim(V_w)$ taken over all $w \in V(T)$ and call it the \emph{mim-width} of $(T, \delta)$. The \emph{mim-width} of $G$ is the minimum mim-width of a branch decomposition of~$G$.
\end{definition}

\begin{definition}[neighborhood equivalence, \cite{BELMONTE201354, BUIXUAN201366}]\label{def:equiv}
  Let $G$ be a graph and let $A \subseteq V(G)$. We say that two subsets $X,Y \subseteq A$ are \emph{neighborhood equivalent with respect to $A$}, denoted by $X \equiv_A Y$, if $N(X) \cap \overline{A} = N(Y) \cap \overline{A}$.
  
It is not hard to see that $\equiv_A$ is an equivalence relation.
  We write $\nec(A)$ for its number of equivalence classes. If $(T, \delta)$ is a branch decomposition of $G$, we denote by $\nec(T, \delta)$ the maximum of $\nec(V_w)$ and $\nec(\overline{V_w})$ over all $w \in V(T)$.
\end{definition}

The following lemma relates maximum induced matchings to neighborhood equivalence.
\begin{lemma}[{\cite[Lemma~1]{BELMONTE201354}}]
  \label{remy}
  For every $n$-vertex graph $G$ and $A \subseteq V(G)$, we have $\mim(A) \leq k$ if and only if, for
  every $S\subseteq A$ there is a $R\subseteq S$ such that $R \equiv_A S$ and $|R|\leq k$.
% \item \label{remy2e} $\nec_d(A) \leq n^{d \cdot \mim(A)}$.
%\item \label{remy2e} $\nec(A) \leq n^{\mim(A)}$.
\end{lemma}

\begin{definition}[Boolean-width]
If $(T, \delta)$ is a branch decomposition of a graph $G$, the \emph{boolean-width} of $(T, \delta)$, denoted by $\boolw(T,\delta)$, is defined as the maximum of $\log( \nec(V_w))$ over all $w\in V(T)$. The \emph{boolean-width} of $G$, denoted by $\boolw(G)$, is the minimum boolean-width of a branch decomposition of~$G$.
\end{definition}

Our results on the boolean-width of $H$-graphs follow from the next result.

\begin{theorem}\label{mim2h}
  Let $H$ be a graph.
  Given any $H$-graph $G$ on $n\geq 2$ vertices, an $H$-representation of $G$ and the corresponding $H$-subdivision, one can compute in polynomial time a caterpillar decomposition $(T, \delta)$ with $\mim(T, \delta) \leq \max\{1, 2\|H\|\}$.
\end{theorem}

\begin{proof}
  We first assume that $H$ is connected, and explain at the end of the proof how we proceed when it is not the case.
  Let $F$ be the subdivision of $H$ in which $G$ can be realized and
  let $\{M_v\}_{v\in V(G)}$ be the intersection representation of $G$. We assume that both $F$ and $\{M_v\}_{v\in V(G)}$ are given as input of the algorithm that we describe now (in addition to $G$).
  Let us arbitrarily fix a branching node $r$ of $F$. Let
  $v_1, \dots, v_n$ be an ordering of $V(G)$ by non-decreasing distance of $M_{v_i}$'s
  to~$r$.

  \begin{claimn}\label{neighrep}
    For every prefix $A$ of $v_1, \dots, v_n$ and every $S \subseteq A$, there is
    a set $R \subseteq S$ of size at most $\max\{1, 2\|H\|\}$ such that 
%PG:
$R \equiv_A S$.
%$N(R) \cap\overline{A} = N(S) \cap \overline{A}$.
  \end{claimn}

  \begin{proof}
  Let $A$ be a prefix of $v_1, \dots, v_n$ and let $S \subseteq A$.
  If there is a vertex $u\in S$ such that $N(u) \cap \overline{A} = N(S) \cap \overline{A}$, we set $R = \{u\}$ and we are done. This includes the case where $G$ is a disjoint union of cliques, which happens for example when $H = K_1$. If such a vertex does not exist, then $H$ has more than one node; since we assume that it is connected, it also has at least one edge.
  Recall that $M_A = \bigcup_{v \in A} M_v$ and similarly for
  $M_{\overline{A}}$ and~$M_S$.
  Let us consider the path $P_e$ corresponding to some edge $e \in E(H)$.
  Let $x_1, \dots, x_p$ be the nodes of $P_e$ in the same order.

  Let $v \in A$ and notice that since, by definition, $H[M_v]$ is connected, the vertex set $M_v \cap V(P_e)$ induces at most two
  connected components in~$P_e$. Indeed if $M_v \cap V(P_e)$ induced more than two connected components, then one of them would not contain any endpoint of $P_e$, and thus this component would not be connected to other nodes of $M_v$ in $H[M_v]$. Let us assume that it induces at least one connected component and let $x_i$ and $x_j$ be the first and last
  nodes (wrt.\ the ordering $x_1, \dots, x_p$) of this component. If $\{x_1,\dots,x_{i-1}\}$
  is disjoint from $M_{\overline{A}}$, we say that $v$ is a \emph{left-protector} of~$P_e$. If $j$ is maximum among all vertices that protect
  the left of $P_e$, then $v$ is a \emph{rightmost left-protector}. (Informally,
  it extends the most to the right.)
  Similarly, $v$ is a \emph{right-protector} when the right of $P_e$
  if $\{x_{j+1}, \dots, x_p\}$ is disjoint from $M_{\overline{A}}$ and
  is a \emph{leftmost right-protector} if $i$ is minimal.

  Let $Z_e$ be a set containing one (arbitrarily chosen) rightmost left-protector and one leftmost right-protector of $e$ if some exist, and
  let $R = \bigcup_{e\in E(H)} Z_e$. Clearly $|R| \leq 2\|H\|$.
  Let us now show that $N(S) \cap \overline{A} \subseteq N(R) \cap \overline{A}$.
  We consider a vertex $u \in N(S) \cap \overline{A}$ and we show that
  it also belongs to $N(R)$.
  Let $v$ be a neighbor of $u$ in $S$. As $u$ and $v$ are adjacent, $M_u$ and $M_v$ have non-empty intersection. Let $e$ be an edge of $H$ such that $M_u$ and $M_v$ meet on $P_e$, i.e.\ $M_u \cap M_{v} \cap V(P_e) \neq \emptyset$. Again, we denote by $x_1, \dots, x_p$ the nodes of~$P_e$.

  \begin{claimn}\label{inprotect}
    Let $w \in A$. If $M_w = \{x_i,\dots, x_{j}\}$ for some $i,j \in \{1, \dots p\}$ with $i\leq j$,
    then one of $\{x_1, \dots x_{i-1}\}$ and $\{x_{j+1}, \dots, x_p\}$
    is disjoint from $M_{\overline{A}}$.
  \end{claimn}
  
  \begin{proof}
    If there are vertices $u,u'$ of $\overline{A}$ such that $M_u$ and
    $M_{u'}$ respectively intersect $\{x_1, \dots x_{i-1}\}$ and
    $\{x_{j+1}, \dots, x_p\}$, then one of $\dist_F(M_u, r)$ and
    $\dist_F(M_{u'}, r)$ is smaller than $\dist_F(M_w, r)$. This
    contradicts the fact that $w \in A$ whereas $u,u' \not \in A$ and proves \autoref{inprotect}.
    \cqed
  \end{proof}

  \begin{claimn}\label{rightleft}
    Let $w \in A$. If $M_w$ intersects $V(P_e)$ then it is a right-protector or a left-protector.
  \end{claimn}

  \begin{proof}
    By definition, if $x_1 \in M_w$ then $M_w$ is a left-protector of $P_e$
    (and symmetrically for the right). The case where $M_w$ contains none of $x_1$ and $x_p$ follows from \autoref{inprotect}.
    \cqed
  \end{proof}

  As $M_u$ intersects $M_v$ on $P_e$, it intersects the vertex set $C$ of
  one component induced by $M_v$ on $P_e$ (recall that
  there are either one or two such components). In the case where
  there are two components, we assume without loss of generality that
  this is the ``left'' one (i.e.\ that with smallest indices).
  In the case where there is one component, we assume that $v$ is a left-protector
  of $P_e$ (according to \autoref{rightleft}, $v$ is a left-protector or a right-protector of $P_e$).
  Observe that in both cases, $v$ is a left-protector of $P_e$.
  Let $z$ be the rightmost left-protector of $P_e$ that belongs to $R$
  and let $x_{k}, \dots x_{k'}$ be the nodes of the corresponding component of
  $P_e[M_z \cap V(P_e)]$ (that is, the component used in the
  definition of left-protector).

  Notice that $C\subseteq \{x_1, \dots, x_{k'}\}$, by
  maximality of $z$ (informally, because it is ``rightmost'').
  As $z$ is a left-protector, $M_u \cap \{x_1, \dots, x_{k-1}\} = \emptyset$.
  Since $M_u$ and $C$ intersect, they intersect in $\{x_k, \dots,
  x_{k'}\}$. Therefore $M_u \cap M_z \neq \emptyset$: $z$ is
  adjacent to $u$. As $z \in R$, we are done.

  This concludes the proof of \autoref{neighrep}.%
\cqed
\end{proof}

We construct a caterpillar decomposition that follows the ordering $v_1,\ldots,v_n$ as follows. We construct a path  $x_1\ldots x_{n}$ and  $n$ vertices $y_1,\ldots,y_n$. Then we make $y_i$ adjacent to $x_i$, for every $i \in \{1, \dots, n\}$. We define $\delta(y_i)=v_i$ for $i\in\{1,\ldots,n\}$. The root is chosen arbitrarily. 
According to \autoref{neighrep} and Lemma~\ref*{remy}, this caterpillar decomposition satisfies $\mim(T, \delta) \leq \max\{1, 2\|H\|\}$. Regarding the running time, we observe that the ordering $v_1, \dots,v_n$ can be found by first labelling the vertices of $F$ with their distance from $r$ obtained by a BFS (in $\Oh(|F| + \|F\|)$ steps) then finding, for each $v \in V(G)$, then minimum label of a vertex in $M_v$ (in $O(\sum_{i=1}^n |M_v|) = O(n|F|)$ steps) and finally sorting these values (in $O(n \log n)$ steps). Overall, the algorithm thus takes polynomial time in the sizes of $G$ and $F$.

We now consider the case when $H$ is not connected.
Then $G$ is the disjoint union of connected graphs $G_1,\dots, G_p$ where for every $i\in \{1, \dots, p\}$, $H$ has a connected component $H'$ such that $G_i$ is an $H'$-graph.
For each $G_i$ we can obtain an ordering $v^i_1, \dots, v^i_{|G_i|}$ as explained above.
Let $T$ be a graph obtained from the path on vertex set
\[x^1_1, \dots, x^1_{|G_1|}, \dots, x^i_1, \dots, x^i_{|G_i|}, \dots, x^p_1, \dots, x^p_{|G_p|}\]
(in this order) by adding a degree one vertex $y^i_j$ adjacent to $x^i_j$, for every $i \in \{1, \dots, p\}$ and every $j \in \{1, \dots, |G_i|\}$. We root $T$ at $x_1^1$.
We also define $\delta(y^i_j)=x^i_j$ for $i$ and $j$ as above. It is easy to check that $(T, \delta)$ is a caterpillar decomposition of $G$ and that
\[
\left (T[x^i_1, \dots, x^i_{|G_i|}, y^i_1, \dots, y^i_{|G_i|}], \delta_{|\{y^i_1, \dots, y^i_{|G_i|}\}} \right )
\]
is a caterpillar decomposition of $G_i$ as defined at the end of the connected case.
(Intuitively, we constructed a caterpillar decomposition of $G$ by attaching the caterpillar decompositions of the $G_i$'s end-to-end.)
As there is no edge in $G$ between vertices of distinct $G_i$'s, we have $\mim_G(V_{x^i_j}) = \mim_{G_i}(V_{x^i_j})$ and $\mim_G(V_{y^i_j}) = \mim_{G_i}(V_{y^i_j})$, for every $i \in \{1, \dots, p\}$ and every $j \in \{1, \dots, |G_i|\}$. Hence the desired bound on the width of $(T, \delta)$ follows from the connected case.

This concludes the proof of \autoref{mim2h}.\qed
\end{proof}

\begin{remark}
  We note that when $H$ is a tree, the bound in \autoref{mim2h} can be improved to $\max \{1, \|H\|\}$. Indeed in this case, in the proof of \autoref{neighrep}, for $v \in A$ and $e \in H$, the vertex set $M_v\cap V(P_e)$ induces at most one component (this follows from the choice of the sequence $v_1, \dots, v_n$). Therefore $|Z_e|\leq 1$ for every $e\in H$ and $|R| \leq \max\{1, \|H\|\}$. The rest of the proof is identical.
\end{remark}
  % 

% The next result follows from the application to the decomposition provided by \autoref{mim2h} of \autoref*{remy}.\eqref{remy2e}, with the fact that $\mim(A) = \mim(\overline{A})$ for every $A \subseteq V(G)$.
% \begin{corollary}\label{necd}
%   Let $H$ be a graph. For every $H$-graph $G$ on $n$
%   vertices whose intersection representation is given, one can compute in
%   polynomial time a caterpillar decomposition $(T, \delta)$ with
%  $\nec_d(T,\delta) \leq n^{d \cdot \max \{1, 2\|H\|\}}$.
% \end{corollary}
The following is immediate.
\begin{corollary}\label{cor:mimh}
  Let $H$ be a graph. Every $H$-graph has mim-width at most $\max \{1, 2\|H\|\}$.
\end{corollary}

From the definition of boolean-width, we also get the next result.
\begin{corollary}\label{boollog}
  Let $H$ be a graph. Every $n$-vertex $H$-graph with $n\geq 2$ has boolean-width at most \[\max \{1, 2\|H\|\} \cdot \log n.\]
\end{corollary}
By choosing $H$ to be a single or double edge, we recover the results of \cite{BELMONTE201354} on the boolean-width of interval and circular-arc graphs, respectively, as special cases of \autoref{boollog}. As proven in the same paper, there is a infinite family of interval graphs with boolean-width $\Omega(\log n)$. Apart of the degenerate case where $H$ is edgeless (in which case $H$-graphs are disjoint unions of cliques), every interval graph is an $H$-graph. This shows that the bound in \autoref{boollog} is tight up to a constant factor.

We now provide algorithmic applications of our results.
Boolean-width has been used \cite{BUIXUAN20115187} to design parameterized algorithms for the problems \textsc{Maximum Weight Independent Set} and \textsc{Minimum Weight Dominating Set}. Later, invariants related to neighborhood equivalence were used in \cite{BUIXUAN201366} as parameters of FPT algorithms for the vast class of locally checkable vertex subset and vertex partitioning problems (\emph{\textsf{LC-VSP}} problems), defined as follows.

\begin{definition}[\cite{BUIXUAN201366}]
  Let $\sigma$ and $\rho$ be finite or co-finite subsets of natural numbers. A subset $S$ of vertices of a graph $G$ is a \emph{$(\sigma, \rho)$-set} of G if
  \[
    \forall v \in V(G),\ |N(v) \cap S| \in
    \begin{cases}
      \sigma & \text{if}\ v \in S \\
      \rho & \text{otherwise}.
    \end{cases}
  \]
  
  A computational problem is \emph{\textsf{LC-VSP}} if it consists in finding a minimum or maximum $(\sigma, \rho)$-set in an input graph, for some $\sigma$ and $\rho$ as above.
\end{definition}

The class of \textsf{LC-VSP} problems include fundamental problems as \textsc{Independent Set}, \textsc{Independent Dominating Set}, \textsc{Total Dominating Set}, and \textsc{Induced Matching}. We refer to \cite{BUIXUAN201366} for several other examples of classic computational problems expressed as \textsf{LC-VSP} problems. The main result of \cite{BUIXUAN201366} is the following. While its original statement deals with the relation of $d$-neighborhood equivalence (an extension of the notion defined in \autoref{def:equiv}), we state it here in terms of mim-width using the direct connection between these two notions given in~\cite[Lemma~2]{BELMONTE201354}.

\begin{theorem}[\cite{BUIXUAN201366}]\label{thlcvsp}
  % For every \textsf{LC-VSP} problem $\Pi$, there are constants $d$ and $q$ such that $\Pi$ can be solved in time $\Oh(n^4 \cdot q \cdot \nec_d(T, \delta)^{3q})$ if a decomposition $(T, \delta)$ of the input is given.
  For every \textsf{LC-VSP} problem $\Pi$, there are constants $d$ and $q$ such that $\Pi$ can be solved in time $\Oh(q \cdot n^{3qd \mim(T, \delta)+4})$ on an input graph of order $n$, if a decomposition $(T, \delta)$ of the input is given.
\end{theorem}

Furthermore, it was recently proved \cite{jaffke2018generalized} that the distance versions of \textsf{LC-VSP} problems (such as $r$-\textsc{Independent Set}, which asks for vertices pairwise at distance at least $r$), are also solvable in polynomial time on graphs of bounded mim-width. We refer to \cite{jaffke2018generalized} for more details.
Regarding problems that are not \textsf{LC-VSP}, Jaffke, Kwon, and Telle obtained polynomial-time algorithms for problems pertaining to induced paths in graphs of bounded mim-width.
\begin{theorem}[\cite{jaffke2017polynomial}]\label{thnlcvsp}
  The problems \textsc{Longest Induced Path}, \textsc{Induced Disjoint Paths}, and, for every graph $J$, the problem \textsc{$J$-Induced Subdivision}\footnote{We refer the reader to \cite{jaffke2017polynomial} for an accurate definition of these problems.} can be solved in time $n^{\Oh(\mim(T, \delta))}$ on an input graph of order $n$, if a decomposition $(T, \delta)$ of the input is given.
\end{theorem}

Composing \autoref{mim2h} with the two aforementioned results, we get the following meta-algorithmic consequences.

\begin{theorem}\label{thm:lc-vsphg}
  Let $H$ be a graph and let $\Pi$ be a (distance) \textsf{LC-VSP} problem.
  Given any $H$-graph, an $H$-representation of it and the corresponding $H$-subdivision $F$, one can solve $\Pi$ in polynomial time.
\end{theorem}

By summing the running times of \autoref{mim2h} (as detailled in its proof) and \autoref{thlcvsp}, we can bound the running time of the algorithm of \autoref{thm:lc-vsphg} by
\[
\Oh(\|F\| + n|F|+ q \cdot n^{6qd\|H\|+4}),
\]
where $q,d>0$ are the constants depending on the problem $\Pi$ given by \autoref{thlcvsp} and $n$ is the order of~$G$.\footnote{We here assumed, for the sake of readability, that $H$ has at least one edge. In the opposite case, $G$ has a simple structure: it is a disjoint union of cliques.}

\begin{theorem}\label{indpa}
  Let $H$ and $J$ be two graphs.
  Given any $H$-graph, an $H$-representation of it and the corresponding $H$-subdivision $F$, one can solve any of
  \textsc{Longest Induced Path}, \textsc{Induced Disjoint Paths}, and \textsc{$J$-Induced Subdivision} in $\Oh(\|F\| + n|F|+ n^{\Oh(\|H\|)})$ time.
\end{theorem}

\section{\texorpdfstring{$H$}{H}-graphs have few minimal separators}
\label{sec:minsep}

Let $G$ be a graph. If $a,b \in V(G)$, we say that $X \subseteq V(G)$ is an \emph{$(a,b)$-separator} if $a$ and $b$ are in distinct connected components of $G \setminus X$. It is a \emph{minimal $(a,b)$-separator} if it is inclusion-wise minimal with this property. A subset of $V(G)$ is a \emph{minimal separator} of $G$ if it is a minimal $(a,b)$-separator for some $a,b \in V(G)$.

The study of minimal separators is an active line of research that found many algorithmic applications
(see e.g.\ \cite{kloks1993computing, Berry1999, doi:10.1137/S0097539799359683, Fomin:2015:LIS:3065779.3065785}). In general, the number of minimal separators of a graph may be as large as exponential in its number of vertices. We prove in this section that in an $H$-graph, this number is upper-bounded by a polynomial (\autoref{minsep}). By combining this finding with meta-algorithmic results of Fomin, Todinca and Villanger \cite{Fomin:2015:LIS:3065779.3065785}, we deduce that a wide class of optimization problems can be solved in polynomial time on $H$-graphs (\autoref{c:oispt}). We complement these results by providing in \autoref{lem:minsepmany} a lower bound on the function of \autoref{minsep}.

\begin{theorem}\label{minsep}
  Let $H$ be a graph. If $G$ is a $H$-graph, it has at most $(2|G| + 1)^{\|H\|} + \|H\|\cdot
(2|G|)^2$ minimal separators.
\end{theorem}

\begin{proof}
  Let $G$ be a $H$-graph.
  Observe that if $H$ is edgeless, then $G$ is a disjoint union of cliques and thus has either only one  minimal separator, the empty set, or none if $G$ is a complete graph. Therefore we may now assume that $H$ has at least one edge.
  
  Let $F$ be a subdivision of $H$ where $G$
  can be represented as the intersection graph of~$\{M_v,\ v \in V(G)\}$.
  For every subset $V \subseteq V(G)$, the \emph{border edges} of $V$
  are the edges of $F$ with one endpoint in $M_V$ and one endpoint in
  $V(F) \setminus M_V$.
  Let $R$ be the union of border edges over $\{M_v,\ v\in V(G)\}$.
  Observe that for every $V \subseteq V(G)$, the set of border edges
  of $V$ is a subset of~$R$.
  For every edge $e \in E(F)$, we set
  \[
    V_e = \{v \in V(G),\ e \subseteq M_v\}
  \]
  and extend this notation to sets $S \subseteq E(F)$ as follows:
  \[
    V_S = \{v \in V(G),\ \exists s \in S,\ s\subseteq M_v\}.
  \]
  Informally, $V_S$ is the set of all vertices of $G$ whose models
  contain some edge of~$S$.
  
  \begin{claimn}\label{cutbord}
    For every minimal separator $X$ in $G$, there is a $S \subseteq R$ such
    that $X = V_S$.
  \end{claimn}

  \begin{proof}
  Let $A,B$ be
  two connected components of $G\setminus X$ such that $N(A) = N(B) =
  X$. As $X$ is an $(A,B)$-separator (i.e.\ $A$ and $B$ are included in the vertex sets of distinct connected components of $G \setminus X$), $M_A \cap M_B = \emptyset$.
  Let $S$ be the set of all border edges of $M_A$ that belong to some
  inclusion-wise minimal path that starts in $A$ and ends in $B$. As noted above, $S
  \subseteq R$. First we show $V_S \subseteq X$. Let $v \in V_S$. That
  is, $v$ is a vertex $G$ such that $M_v \supseteq s$ for some $s \in
  S$.
  Then $M_v$ contains both endpoints of $s$, one of which belongs
  to $M_A$. The vertex $v$ is adjacent to $A$ but does not belong to
  $A$ (as $M_v$ contains a vertex of $\overline{A}$), so it has to belong to the separator $X$. Therefore, $V_s \subseteq X$.
  Now we show $X \subseteq V_S$. Let $x \in X$. By definition, $x$ has
  a neighbor in both $A$ and $B$. Therefore, $M_x$ meets both $M_A$
  and $M_B$. As $M_x$ induces a connected subgraph of $F$ and $M_A$
  is disjoint from $M_B$, it contains an edge $s \in E(F)$ with one
  endpoint in $M_A$ and the other in $V(F) \setminus M_A$. Then $s$ is
  a border edge of $M_A$ in a minimal path from $A$ to $B$: $x \in
  V_S$. Hence $X = V_S$.%
  \cqed
  \end{proof}
  
  From \autoref{cutbord} we can already deduce that the number of
  minimal separators of $G$ is at most the number of subsets of
  $R$. In order to obtain better bounds, we need other
  observations.

  \begin{claimn}\label{2edges}
    For every $V \subseteq V(G)$ such that $M_V$ induces a connected
    subgraph of $F$, and every $e \in E(H)$, the set $M_V$ has
    at most two border edges in $E(P_e)$. Hence, $|R| \leq 2|G| \cdot \|H\|$.
  \end{claimn}
  \begin{proof}
    Follows from the fact that $F[M_V]$ is connected.%
    \cqed
  \end{proof}

  \begin{claimn}
   For every minimal separator $X$ of $G$, if $S\subseteq R$ is the
   subset of edges of $F$ defined in the proof of \autoref{cutbord}, then
   \begin{itemize}
   \item either $|S \cap E(P_e)| \leq  1$ for every $e \in E(H)$; 
   \item or $|S| = 2$ and $S \subseteq E(P_e)$ for some $e \in E(H)$.
   \end{itemize}
 \end{claimn}
 \begin{proof}
   Let $A$ and $B$ be as in the  proof of \autoref{cutbord}. According
   to \autoref{2edges} and as $S$ is a subset of the border edges of $M_A$,
   we deduce $|S \cap E(P_e)| \leq 2$ for every $e\in E(H)$. Let us assume that $|S \cap
   E(P_e)| = 2$ for some $e\in E(H)$. Let $u,u'$ and $v,v'$ be the endpoints of the two
   edges shared by $S$ and $E(P_e)$, respectively and in this order on the
   path. Then the model of one of $A$ and $B$ has its vertices in the
   subpath $Q$ of $P_e$ delimited by $u'$ and $v$. Indeed, both $\{u,u'\}$
   and $\{v, v'\}$ and have an endpoint that does not belong to
   $M_A$. As $M_A$ induces a connected subgraph of $F$, either these
   endpoints are $u'$ and $v$ (intuitively, the exterior endpoints) or they are $u$,
   $v'$ (the interior endpoints). In the first case $M_A\subseteq E(Q)$ and in the second one,
   $M_B \subseteq E(Q)$. From the definition of $R$, we can then conclude that $S \subseteq E(P_e)$ and we get $|S|=2$.%
   \cqed
 \end{proof}

 Therefore, for every minimal separator $X$ of $G$, there is a set $S
 \subseteq R$ such that:
 \begin{enumerate}
 \item either $|S \cap E(P_e)| \leq 1$ for every $e \in E(H)$;\label{intone}
 \item or $|S| =2$ and $S \subseteq E(P_e)$ for some $e \in E(H)$;\label{inttwo}
 \end{enumerate}

 In order to upper-bound the number of possible minimal separators of
 $G$, it suffices to upper-bound the number of sets $S \subseteq
 R$ that satisfy one of the two conditions above.
As noted in \autoref{2edges}, for every $e \in E(H)$ we have $R \cap
E(P_e) \leq 2|G|$. Hence there are at most $(2|G|)^2$ possible choices of set
$S$ that satisfy \eqref{inttwo} for each $e \in E(H)$. We deduce that
there are at most $\|H\|\cdot (2|G|)^2$ distinct sets $S \subseteq R$
satisfying~\eqref{inttwo}. Let us now consider sets $S \subseteq R$ that satisfy~\eqref{intone}. For every $e\in
E(H)$, either $S$ contains one of the $2|G|$ edges of $R \cap E(P_e)$ or
it does not contain any of them. This makes $2|G|+1$ possible
choices for each $e \in E(H)$, and $(2|G| + 1)^{\|H\|}$ in total.
Consequently, $G$ has at most $(2|G| + 1)^{\|H\|} + \|H\|\cdot
(2|G|)^2$ minimal separators. This proves~\autoref{minsep}.%
\qed
\end{proof}

For every $r\in \N$, let $\theta_r$ be the graph with 2 vertices and
$r$ parallel edges.
The following shows that the exponential contribution of $\|H\|$ in \autoref{minsep} cannot be avoided.

\begin{lemma}\label{lem:minsepmany}
For every $r\in \N$, there is a $\theta_r$-graph $G$ with at least
$\left (\frac{|G|-2}{r}\right )^r$ minimal separators.
\end{lemma}

\begin{proof}
  Let $G$ be the graph obtained from $\theta_r$ by subdividing $k$
  times each edge (see \autoref{fig:minsep} for an example with $r=4$).
  Then $G$ is a $\theta_r$-graph and $|G| = kr+2$. Notice that any choice
  of $r$ subdivision nodes, each corresponding to a different edge
  of $\theta_r$, gives a distinct minimal separator of $G$. Hence $G$
  has at least $k^r = \left (\frac{|G|-2}{r}\right )^r$ minimal
  separators.%
  % exact number : $\left (\frac{|G|-2}{r}\right )^r + \binom{k-2}{2} - (k-2)
  % + 2(k-3) + 1$
  \qed
\end{proof}

\begin{figure}[h]
  \centering
  \begin{tikzpicture}[every node/.style = black node]
    \draw (0,0) node (a) {} (4,0) node (b) {};
    \draw
    (a) --++(1, 0.45) node {} -- ++(1,0) node{} -- ++(1,0) node{} node[midway, normal, fill = white] {$\dots$} -- (b)
    (a) --++(1, 0.15) node {} -- ++(1,0) node{} -- ++(1,0) node{} node[midway, normal, fill = white] {$\dots$} -- (b)
    (a) --++(1, -0.15) node {} -- ++(1,0) node{} -- ++(1,0) node{} node[midway, normal, fill = white] {$\dots$} -- (b)
    (a) --++(1, -0.45) node {} -- ++(1,0) node{} -- ++(1,0) node{} node[midway, normal, fill = white] {$\dots$} -- (b);
    \draw[decorate,decoration={brace,amplitude=5pt}] (3.25,-0.65) -- (0.75, -0.65) node[midway, normal, anchor = north, yshift = -0.25cm] {$k$};
  \end{tikzpicture}
  \caption{A $\theta_4$-graph with at least $k^4$ minimal separators.}
  \label{fig:minsep}
\end{figure}

Our results on minimal separators have algorithmic consequences.
Let $t \in \N$ and let $\mathcal{P}$ be a boolean function depending on a graph and a subset of its vertices. (More formally, $\mathcal{P}(G,X)$ is a boolean value, for every graph $G$ and $X \subseteq V(G)$.)
We consider the following generic problem described in~\cite{Fomin:2015:LIS:3065779.3065785}.

\begin{center}
\fbox{\begin{minipage}{0.95\textwidth}
    \noindent{\textsc{Optimal Induced Subgraph for $\mathcal{P}$ and $t$}, $\textsc{OIS}(\mathcal{P}, t)$ for short}
    \begin{description}
    \item[Input:] A graph $G$
    \item[Task:] Find sets $X\subseteq Y \subseteq V(G)$ such that $X$ is of maximum size, the induced
    subgraph $G[Y]$ is of treewidth at most $t$, and $\mathcal{P}(G[Y], X)$ is true.
    \end{description}
\end{minipage}}
\end{center}

For various choices of $\mathcal{P}$ and $t$, this generic problem corresponds to natural families of optimization meta-problems like \textsc{$\mathcal{F}$-minor-deletion} (where $\mathcal{F}$ is a class of graphs containing at least one planar graph) whose goal is to delete a minimum number of vertices in order to get an $\mathcal{F}$-minor free graph\footnote{In fact, \textsc{Optimal Induced Subgraph for $\mathcal{P}$ and $t$} corresponds to the dual equivalent problem of \textsc{$\mathcal{F}$-minor-deletion}, which asks for a largest $\mathcal{F}$-minor free subgraph of the input.} and \textsc{Independent $\mathcal{F}$-packing} (where $\mathcal{F}$ is a class of connected graphs), which asks for a maximum number of disjoint copies of graphs in $\mathcal{F}$ as pairwise independent subgraphs of the input.
Fomin, Todinca, and Villanger proved that when the property $\mathcal{P}$ can be expressed in Counting Monadic Second Order logic (\textsf{CMSOL}, see \cite{Fomin:2015:LIS:3065779.3065785}), the above problem can be easily solved on classes of graphs that have a polynomial number of minimal separators.

\begin{theorem}[\cite{Fomin:2015:LIS:3065779.3065785}]\label{th:minsepmalgo}
    For any fixed $t\in \N$ and \textsf{CMSOL} property $\mathcal{P}$, $\textsc{OIS}(\mathcal{P}, t)$ is solvable on an $n$-vertex graph with $s$ minimal separators in time 
    $\Oh(s^2\cdot n^{t+4} \cdot f(t, \mathcal{P}))$, for some function  $f$ of $t$ and $\mathcal{P}$ only.
    %In particular, the problem is 
    %solvable in polynomial time for classes of graphs whose number of minimal separator is upper-bounded by a polynomial function of their order.
  \end{theorem}

  We deduce that $\textsc{OIS}(\mathcal{P}, t)$ can be solved in polynomial time in $H$-graphs:
  \begin{corollary}\label{c:oispt}
    Let $H$ be a graph. For any fixed $t\in\N$ and \textsf{CMSOL} property $\mathcal{P}$, $\textsc{OIS}(\mathcal{P}, t)$ can be solved on an $n$-vertex $H$-graph in time $n^{\Oh(\|H\|+ t+4)} \cdot f(t, \mathcal{P})$, for some function  $f$ of $t$ and $\mathcal{P}$ only.
  \end{corollary}
  
\section{Parameterized complexity of basic problems for \texorpdfstring{$H$}{H}-graphs}\label{sec:basic}
\label{sec:pc}

In this section we investigate the parameterized complexity of some basic graph problems for $H$-graphs: \textsc{Dominating Set}, \textsc{Independent Set} and \textsc{Clique}. First, in \autoref{sec:hard}, we show that  \textsc{Dominating Set} and \textsc{Independent Set} are \classW{1}-hard when parameterized by the solution size and the size of $H$. In \autoref{sec:ds-tree}, we show that  \textsc{Dominating Set} is \classFPT when parameterized by the number of vertices of $H$ if $H$ is a tree. In fact, we show a more general result by proving that \textsc{Dominating Set} is \classFPT for chordal graphs if the problem is parameterized by the \emph{leafage} of the input graph, that is, by the minimum number of leaves in a clique tree for the input graph. This result is somehow tight since \textsc{Dominating Set} is well-known to be \classW{2}-hard for split graphs when parameterized by the solution size~\cite{RamanS08}. Recall also that \textsc{Independent Set} is polynomial-time solvable for chordal graphs~\cite{Gavril72,Golumbic04} and, therefore, for $H$-graphs if $H$ is a tree. Finally, in \autoref{sec:clique}, we show that \textsc{Clique} admits a polynomial kernel when parameterized by the solution size and the size of $H$, in the case where the representation is given.

\subsection{Hardness of Independent Set and Dominating Set on \texorpdfstring{$H$}{H}-graphs}\label{sec:hard}
In this section we prove \classW{1}-hardness of \textsc{Dominating Set} and \textsc{Independent Set} for $H$-graphs (\autoref{thm:ds-is}).  Recall that \textsc{Dominating Set} and \textsc{Independent Set}, given a graph $G$ and a positive integer $k$, ask whether $G$ has a dominating set of size at most $k$ and independent set of size at least $k$ respectively. To show hardness, we reduce from the \textsc{Multicolored Clique} problem. This problem, given a graph $G$ with a $k$-partition of its vertex set $V_1,\ldots,V_k$, asks whether $G$ has a $k$-clique with exactly one vertex in each $V_i$ for $i\in\{1,\ldots,k\}$. The problem is well-known to be \classW{1}-complete when parameterized by~$k$~\cite{FellowsHRV09,Pietrzak03}.

\begin{theorem}\label{thm:ds-is}
\textsc{Dominating Set} and \textsc{Independent Set} are \classW{1}-hard for $H$-graphs when parameterized by $k+\|H\|$ and the hardness holds even if an $H$-representation of $G$ is given.
\end{theorem}

\begin{proof}
First, we show the \classW{1}-hardness for \textsc{Independent Set} and then explain how to modify the reduction for \textsc{Dominating Set}. The reduction is from \textsc{Multicolored Clique}. 

Let $(G,V_1,\ldots,V_k)$ be an instance of \textsc{Multicolored Clique}. We assume that $k\geq 2$ and $|V_i|=p$ for $i\in\{1,\ldots,k\}$. The second assumption can be made without loss of generality because we always can add isolated vertices to the sets $V_1,\ldots,V_k$ to ensure that they have the same size. Denote by $v_1^i,\ldots,v_p^i$  the vertices of $V_i$ for $i\in\{1,\ldots,k\}$.

\begin{figure}[ht]
\centering
\scalebox{0.75}{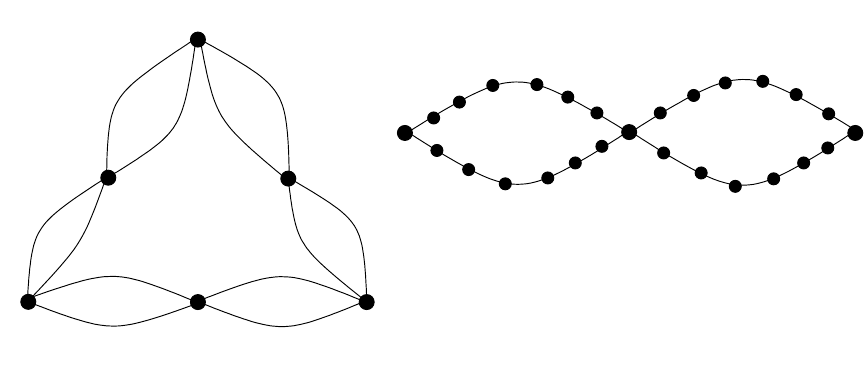}
\caption{The construction of $H$ for $k=3$ and the subdivision of the edges of $H$.\label{fig:h}}
\end{figure}

We construct the multigraph $H$ as follows (see \autoref{fig:h} a)).
\begin{enumerate}[(i)]
\item \label{constrind1} Construct $k$ nodes $u_1,\ldots,u_k$.
\item \label{constrind2} For every $i,j \in \{1, \dots, k\}$ with $i<j$, construct a node $w_{i,j}$ and two pairs of parallel edges $u_iw_{i,j}$ and $u_jw_{i,j}$.
\end{enumerate}
Note that $|H|=k(k+1)/2$ and $\|H\|=2k(k-1)$.

Then we construct the subdivision $H'$ of $H$ obtained by subdividing each edge $p$ times. We denote the subdivision nodes for the 4 edges of $H$ constructed for each $i,j \in \{1, \dots, k\}$ with $i<j$ in \eqref{constrind2} by $x_1^{(i,j)},\ldots,x_p^{(i,j)}$, $y_1^{(i,j)},\ldots,y_p^{(i,j)}$, $x_1^{(j,i)},\ldots,x_p^{(j,i)}$ and $y_1^{(j,i)},\ldots,y_p^{(j,i)}$ as it is shown in \autoref{fig:h} b). To simplify notations, we assume that $u_i=x_0^{(i,j)}=y_0^{(i,j)}$, $u_j=x_0^{(j,i)}=y_0^{(j,i)}$ and $w_{i,j}=x_{p+1}^{(i,j)}=y_{p+1}^{(i,j)}=x_{p+1}^{(j,i)}=y_{p+1}^{(j,i)}$.

\begin{figure}[ht]
\centering
\scalebox{0.75}{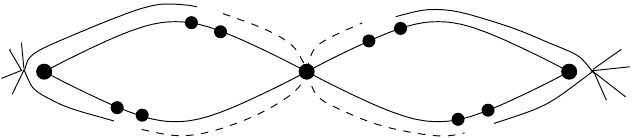}
\caption{The construction of $G'$.\label{fig:g}}
\end{figure}

Now we construct the $H$-graph $G'$ by defining its $H$-representation $\mathcal{M}=\{M_v\}_{v\in V(G')}$ where the model of each vertex is a connected subset of $V(H')$ (see \autoref{fig:g}). Recall that $G$ is the graph of the original instance of \textsc{Multicolored Clique}.
\begin{enumerate}[(i)]
\item For each $i\in \{1,\ldots,k\}$ and $s\in\{1,\ldots,p\}$, construct a vertex $z_s^i$ with the model 
  \[
    M_{z_s^i}=\bigcup_{j\in\{1,\ldots,k\},j\neq i}\left \{ \left \{x_0^{(i,j)},\ldots,x_{s-1}^{(i,j)} \right \}\cup \left \{y_0^{(i,j)},\ldots, y_{p-s}^{(i,j)}\right\}\right \}.
  \]
\item For each edge $v_s^iv_t^j\in E(G)$, $s,t\in\{1,\ldots,p\}$ and $i,j \in \{1, \dots, k\}$ with $i<j$, construct a vertex $r_{s,t}^{(i,j)}$ with the model 
  \begin{align*}
    M_{r_{s,t}^{(i,j)}}= & \phantom{{}\cup{} {}} \left \{x_{s}^{(i,j)},\ldots,x_{p+1}^{(i,j)} \right\}\\
                         &\cup \left \{y_{p-s+1}^{(i,j)},\ldots,y_{p+1}^{(i,j)} \right \}\\
                         &\cup \left \{x_{t}^{(j,i)},\ldots,x_{p+1}^{(j,i)} \right \}\\
                         &\cup \left \{y_{p-t+1}^{(j,i)},\ldots,y_{p+1}^{(j,i)} \right \}.
  \end{align*}
\end{enumerate}
Note that the neighborhood of $r^{(i,j)}_{s,t}$ is $(V_i \cup V_j) - {v^i_s,v^j_t}$.
Finally, we define $k'=k(k+1)/2$.
We claim that $(G,V_1,\ldots,V_k)$ is a yes-instance of \textsc{Multicolored Clique} if and only if $G'$ has an independent set of size $k'$.
The proof is based on the following crucial property of our construction, that can be easily checked.

\begin{claimn}\label{claim:star}
  For every $i,j \in \{1, \dots, k\}$ with $i<j$, a vertex $z_h^i\in V(G')$ (a vertex $z_h^j\in V(G')$) is not adjacent to a vertex $r_{s,t}^{(i,j)}\in V(G')$ corresponding to the edge $v_s^iv_t^j\in E(G)$ if and only if $h = s$ ($h = t$, respectively).
\end{claimn}
\cqed

We now show that $G'$ has an independent set of size $k'$ if $G$ has a clique of size $k$, and vice-versa.
Let $\{v_{h_1}^1,\ldots,v_{h_k}^k\}$ be a clique of $G$. Consider the set 
\[
  I = \left \{z_{h_1}^1,\ldots, z_{h_k}^k \right \} \cup \left \{ r_{h_i,h_j}^{(i,j)}\mid 0\leq i<j\leq k \right \}
\]
of vertices of $G'$.
It is straightforward to verify using \autoref{claim:star} that $I$ is an independent set of size $k'$ in~$G'$.

Suppose now that $G'$ has an independent set $I$ of size $k'$. For each $i\in\{1,\ldots,k\}$, the set $Z_i=\left \{z_{h}^i\mid 1\leq h\leq p \right \}$ is a clique of $G'$, and for each $i,j \in \{1, \dots, k\}$ with $i<j$,
the set

\[
  R_{i,j}=\left \{r_{s,t}^{(i,j)}\mid 1\leq s,t\leq p,v_s^iv_t^j\in E(G)\right \}
\]
is also a clique of $G'$.  Since all these $k+\binom{k}{2}=k(k+1)/2=k'$ cliques form a partition of $V(G')$, we have that for each 
$i\in\{1,\ldots,k\}$, there is a unique $z_{h_i}^i\in Z_i\cap I$, and for every $i,j \in \{1, \dots, k\}$ with $i<j$, there is a unique $r_{s_i,s_j}^{(i,j)}\in R_{i,j}\cap I$. Since $r_{s_i,s_j}^{(i,j)}$ is not adjacent to $z_{h_i}^i$ and $z_{h_j}^j$, we obtain that $s_i=h_i$ and $s_j=h_j$ by \autoref{claim:star}. It implies that $v_{h_i}^iv_{h_j}^j\in E(G)$. Since it holds for every $i,j \in \{1, \dots, k\}$ with $i<j$, $\left \{v_{h_1}^1,\ldots,v_{h_k}^k \right\}$ is a clique in $G$.

This completes the \classW{1}-hardness proof for \textsc{Independent Set}. Now we explain how we modify our proof to show the \classW{1}-hardness of \textsc{Dominating Set}.
This time we do not reduce from \textsc{Multicolored Clique} but from the \textsc{Multicolored Independent Set} problem that, given a graph $G$ with a $k$-partition of its vertex set $V_1,\ldots,V_k$, asks whether $G$ has a independent set of size  $k$ with exactly one vertex in each $V_i$ for $i\in\{1,\ldots,k\}$. Clearly, the \classW{1}-completeness of \textsc{Multicolored Clique} parameterized by $k$~\cite{FellowsHRV09,Pietrzak03} immediately implies the same for \textsc{Multicolored Independent Set}. 

Let $(G,V_1,\ldots,V_k)$ be an instance of \textsc{Multicolored Independent Set}. We assume without loss of generality that $k\geq 2$ and $|V_i|=p$ for $i\in\{1,\ldots,k\}$. As before, denote by $v_1^i,\ldots,v_p^i$  the vertices of $V_i$ for $i\in\{1,\ldots,k\}$.
We construct the same multigraph $H$ and its subdivision $H'$ as above. We construct the $H$-graph $G''$ from the graph $G'$ constructed above by adding $k$ new vertices $d_1,\ldots,d_k$ with the models $M_{d_i}=\{u_i\}$ for $i\in\{1,\ldots,k\}$.

We show that $(G,V_1,\ldots,V_k)$ is a yes-instance of \textsc{Multicolored Independent Set} if and only if $G''$ has a dominating set  of size $k$.

Suppose that $\left \{v_{h_1}^1,\ldots,v_{h_k}^k \right\}$ is an independent set of $G$. Consider the set $D=\left \{z_{h_1}^1,\ldots, z_{h_k}^k \right\}$. By \autoref{claim:star} and the construction of $G''$, we obtain that $D$ is a dominating set of $G''$.

Let now $D$ be a dominating set of $G''$ with $|D|=k$. Note that each vertex $d_i$ is adjacent only to the vertices of the set $Z_i=\left \{z_{h}^i\mid 1\leq h\leq p \right \}$ for $i\in\{1,\ldots,k\}$. It implies that for every $i\in\{1,\ldots,k\}$,
\[
  D\cap (Z_i\cup\{d_i\})\neq\emptyset .
\]
Since $Z_i$ is a clique, we can assume without loss of generality that $D\cap Z_i\neq\emptyset$ as, otherwise, we can replace $d_i$ in $D$ by an arbitrary vertex of $Z_i$. Since $Z_i\cap Z_j=\emptyset$ if $i\neq j$, we conclude that $D$ contains a unique vertex from each $Z_i$ and no other vertices. Let  $D= \left \{z_{h_1}^1,\ldots,z_{h_k}^k \right \}$. We claim that $I = \left \{v_{h_1}^1,\ldots,v_{h_k}^k \right \}$ is an independent set of $G$. To obtain a contradiction, assume that $v_{h_i}^iv_{h_j}^j\in E(G)$ for some $i,j \in \{1, \dots, k\}$ where $i<j$. Consider the vertex $r_{h_i,h_j}^{(i,j)}$ of $G''$. By \autoref{claim:star}, $r_{h_i,h_j}^{(i,j)}$ is adjacent neither to $z_{h_i}^i$ no $z_{h_j}^j$. Because $r_{h_i,h_j}^{(i,j)}$ is not adjacent to $z_{h_s}^s$ for any $s\in\{1,\ldots,k\}$ such that $s\neq i,j$, we have that $r_{h_i,h_j}^{(i,j)}$ is not dominated by $D$. This contradiction shows the claim and concludes the proof of \autoref{thm:ds-is}.
\qed
\end{proof}

Recall that we proved in \autoref{mim2h} that for every fixed $H$, every $H$-graph has mim-width at most $2\|H\|+1$. We deduce from the negative results above the following corollary.
\begin{corollary}
  \textsc{Dominating Set} and \textsc{Independent Set} are \classW{1}-hard when parameterized by the solution size plus the mim-width of the input.
\end{corollary}

We note that the construction in the proof of \autoref{thm:ds-is} has been adapted in \cite{jaffke17noteKT} to show that the \textsc{Feedback Vertex Set} problem is \classW{1}-hard on $H$-graphs when parameterized by the solution size plus the number of edges of~$H$.

\subsection{Dominating Set for  \texorpdfstring{$T$}{T}-graphs}\label{sec:ds-tree}
In this section we show that \textsc{Dominating Set} is \classFPT for chordal graphs if the problem is parameterized by the leafage (hereafter defined) of the input graph. We stress that our algorithm does not require the intersection representation of the input graph to be given.

Let $G$ be a graph. As it is standard, we say that $u\in V(G)$ \emph{dominates} $v\in V(G)$ if $v\in N_G[u]$ and $u$ dominates a set $W\subseteq V(G)$ if every vertex of $W$ is dominated by $u$. Respectively, a set $D\subseteq V(G)$ dominates $W\subseteq V(G)$ if every vertex of $W$ is dominated by some vertex of $D$.

Let $G$ be a graph. Let $\mathcal{K}$ be the set of (inclusion-wise) maximal cliques of $G$ and let $\mathcal{K}_v\subseteq \mathcal{K}$ be the set of maximal cliques containing $v\in V(G)$. A tree $T$ whose node set is $\mathcal{K}$ such that each $\mathcal{K}_v$ (for $v \in V(G)$) induces a subtree of $T$ is called a \emph{clique tree} of~$G$. It is well-known~\cite{Gavril74} that $G$ is a chordal graph if and only if $G$ has a clique tree $T$. Moreover, if $T$ is a clique tree of $G$, then $G$ is an intersection graph of subtrees of $T$, that is, $G$ is a $T$-graph. Conversely, if $G$ is a $T$-graph, then there is a clique tree $T'$ of $G$ where the number of leaves of $T'$ is at most the number of leaves of $T$.
Note that a clique tree of a chordal graph is not necessarily unique. For a connected chordal graph $G$, the \emph{leafage} $\ell(G)$ of $G$ is the minimum number of leaves in tree $T$ such that $G$ is a $T$-graph \cite{lin1998leafage} which is also, by the above remarks, the minimum number of leaves in a clique tree of $G$.
It was shown by Habib and Stacho in~\cite{HabibS09} that the leafage of a connected chordal graph $G$ can be found in polynomial time. Their algorithm also constructs a corresponding clique tree $T$ with the minimum number of leaves. In other words,  given a connected chordal graph $G$, we can construct in polynomial time a clique tree $T$ with $\ell(G)$ leaves and a $T$-representation of $G$. Also, if one is given a graph $G$ that is promised to be a $T$-graph, then one can produce a clique tree $T'$ of $G$ in polynomial time where the number of leaves of $T'$ is at most the number of leaves of $T$. 

To solve  \textsc{Dominating Set}, we are going to use a dynamic programming algorithm over a clique tree $T'$ of the input graph $G$. However, to do it, we have to treat the vertices of $G$ whose models contain branching nodes of $T'$ in a special way. The vertices of other type, that is, the vertices whose models  contain only subdivision nodes of $T'$ on each path corresponding to an edge of $T$, induce an interval  graph and a minimum dominating set can be selected by a well-known greedy procedure (see, e.g., \cite{Golumbic04}). If the model of a vertex $v$ contains branching nodes, then this vertex can dominate various vertices whose models are in different parts of $T'$ that can be far away from each other and, symmetrically, such a vertex can be dominated by vertices with models that are in different parts of $T'$.  To overcome these difficulties, we show that it is possible to upper bound the number of these vertices in a minimum dominating set. This allows us to guess the structure of models of the vertices in a minimum dominating sets with respect to branching vertices in them. 
Furthermore, we  apply some reduction rules to the input graph $G$ and $T'$ to simplify the models of the vertices of $G$. More precisely, we obtain a representation such that each model contains at most one branching node. We are paying for these reductions by switching to a special labeled variant of \textsc{Dominating Set} called  \textsc{Dominating Set Extension}. Nevertheless, the obtained representation has \emph{local} models and we can use it to construct a dynamic programming algorithm.

Let $T$ be a tree and let $G$ be a connected $T$-graph with its $T$-representation $\mathcal{M}=\{M_v\}_{v\in V(G)}$ with respect to a subdivision  $T'$ of $T$. 
For every non-empty $Q \subseteq V(T)$, we say that $v\in V(G)$ is a $Q$-vertex  if $M_v\cap V(T)=Q$. If $Q=\{u\}$, we write $u$-vertex instead of $\{u\}$-vertex. 
Also, we denote the set of $Q$-vertices by $V_G(Q)$ and $V_G(u)$ if $Q=\{u\}$. We also denote by $V_G(T)$ the set of all $Q$-vertices of $G$ for every non-empty $Q \subseteq V(T)$. In other words, these are the vertices of $G$ whose models contain nodes of~$T$.
For every $e\in E(T)$,  $v\in V(G)$ is an $e$-vertex if $M_v$ contains only subdivision nodes of $T'$ from the path in $T'$ corresponding to $e$ in $T$. 
The set of $e$-vertices is denoted by $V_G(e)$.

We need the following lemma that allows us to upper bound the number of vertices in a minimum dominating set whose models contain given nodes of $T$.

\begin{figure}[ht]
\centering
\scalebox{0.75}{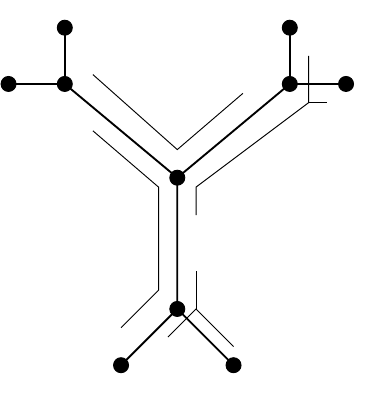}
\caption{The construction of $X=\{x_1,x_2,x_3\}$ for $z_1,z_2,z_3,z_4\in D$; $N_T(X)=\{v_1,\ldots,v_5\}$, $u_1=z_1$, $u_2=z_2$, $u_3=z_4$, $u_4=u_5=z_3$, $u_{x_1x_2}=z_2$ and $u_{x_2x_3}=z_3$. 
\label{fig:lem:upper-a}}
\end{figure}

\begin{lemma}\label{lem:upper-a}
Let $\|T\|\geq 2$ and
let $D$ be a minimum dominating set of $G$. 
Let  also $X\subseteq V(T)$ be a connected 
%an inclusion maximal 
 set of nodes of $T$ such that
\begin{enumerate}[(i)]
\item for every $x\in X$, there is $u\in D$ with $x\in M_u$; and
\item for every $xy\in E(T)$ with $x,y\in X$, there is 
$u\in D$ with $x,y\in M_u$
\end{enumerate}
(see \autoref{fig:lem:upper-a}).
 Then the set $U=\{u\in D\mid X\cap  M_u\neq\emptyset\}$ contains at most  $|N_T[X]|-1$ vertices.   
\end{lemma}

\begin{proof}
%Clearly, if $\max\{|N_T[X]|-1,1\}=1$, then $|T|=1$  and the claim is trivial. Let $|T|\geq 2$.  
Denote by $v_1,\ldots,v_s$ the nodes of  $N_T(X)$.
To obtain a contradiction, assume that $|U|\geq |N_T[X]|=s+|X|$. For each $i\in\{1,\ldots,s\}$, let $u_i\in U$ be a vertex such that the distance between $M_{u_i}$ and $v_i$ in $T'$ is minimum, and   for each $xy\in E(T)$ with $x,y\in X$, let $u_{xy}\in U$ be an arbitrary vertex with $x,y\in M_{u_{xy}}$ (see \autoref{fig:lem:upper-a}).
Let $U'=\{u_1,\ldots,u_\ell\}\cup\{u_{xy}\mid x,y\in X,xy\in E(T)\}$. Note that $|U'|\leq s+\|T[X]\|=s+|X|-1=|N_T[X]|-1$.
Since $X$ induces a subtree of $T$, we have that
$\cup_{u\in U}M_u\subseteq \cup_{u\in U'}M_u$.   This immediately implies that $D'=(D\setminus U)\cup U'$ 
is a dominating set of $G$ contradicting the minimality of~$D$.%
\qed
\end{proof}

In particular, since $|N_T(X)|$ is at most the number of leaves $\ell$, we have that $|U|\leq |X|+\ell-1$. Notice also that $|N_T[X]|-1=\|T[N_T[X]]\|$. 

The next lemma gives an upper bound for the number of vertices in a minimum dominating set whose models contain nodes of $T$ (a very similar bound was given in \cite[Lemma~13]{ChaplickTVZ16}). 

\begin{lemma}\label{lem:upper-b} 
Let $D$ be a minimum dominating set of $G$ and let $\|T\|\geq 1$. Then  $|D\cap V_G(T)|\leq 2|T|-2$.
\end{lemma}

\begin{proof}
%If $|T|=1$, then the claim is straightforward. Let $\|T\|\geq 1$.
Let $D$ be a minimum dominating set of $G$. Consider the set $W$ of nodes of $T$ that are included in the models of the vertices of $D$.  
Let $X_1,\ldots,X_r$ be the partition of $W$ into inclusion maximal connected subsets such that for each $i\in\{1,\ldots,r\}$ and adjacent $x,y\in X_i$, there is $u\in D$ with $x,y\in M_u$. By \autoref{lem:upper-a} and the fact that each edge of $T$ belongs to at most two subtrees $T[N_T[X_i]]$, 
%\begin{align*}
%  |D\cap V_G(T)| &\leq \sum_{i=1}^r(|X_i|+|N_T(X_i)|)\\
%                 &\leq |W|+ \sum_{i=1}^r|N_T(X_i)|\\
%                 &\leq |T|+2\|T\|\\
%                 &=3|T|-2.  
$$ |D\cap V_G(T)| \leq \sum_{i=1}^r\|T[N_T(X_i)]\|\leq 2\|T\|= 2|T|-2.$$
%\end{align*}
\qed
\end{proof}

For an edge $e\in E(T)$, we say that $G'$ is obtained \emph{by contracting $e$ in $T$} if $G'$ is the $(T/e)$-graph with the model obtained as follows: 
\begin{enumerate}[(i)]
\item contract $xy$ in $T$ and, respectively, the $(x,y)$-path $P$ in $T'$, and denote the node obtained from $x$ and $y$ by $z$,
\item delete all $e$-vertices of $G$, 
\item for each remaining vertex $u\in V(G)$, delete from $M_u$ the subdivision nodes of $P$ and replace $x$ and $y$ by $z$ if at least one of these nodes is in $M_u$.
\end{enumerate}
Note that $V(G')\subseteq V(G)$ and $G[V(G')]$ is a subgraph of $G'$ but not necessarily induced since two vertices of $G'$ that are not adjacent in $G$ could be adjacent in $G'$.

Consider a \emph{coloring} $c\colon V_G(T)\rightarrow \{1,\ldots,2^{|T|}\}$ such that for $u,v\in V_G(T)$, $c(u)=c(v)$ if and only if $u$ and $v$ are $Q$-vertices for the same $Q \subseteq V(T)$. 
The next lemmas are used to simplify the models of vertices of $G$ by contracting edges of $T$.

\begin{lemma}\label{lem:contraction}
Let $C\subseteq c(V_G(T))$ and let
\[
  A=\{xy\in E(T)\mid x,y\in M_u\text{ for some~}u\in V_G(T)\text{~such that~} c(u)\in C\}.
\]
Further, let $G'$ be the graph obtained from $G$ by iteratively contracting edges of $A$ in $T$.  Then for any set $D\subseteq V(G)$, $D$ is a minimum dominating set of $G$ satisfying the condition $C= c(D\cap V_G(T))$ if and only if $D$ is a minimum dominating set of $G'$ satisfying the same condition  $C=c(D\cap V_{G}(T))$.
\end{lemma}

\begin{proof}
Before we start proving the lemma, observe that $V_G(T)\subseteq V(G')$ and that $V_G(T)$ is the set of $Q$-vertices of $G'$ for non-empty subsets $Q$ of the set of nodes of $T/A$. Note also that for $G'$ and $T/A$, the coloring $c$ does not necessarily have the property that if $u$ and $v$ are $Q$-vertices of the same $Q$, then $c(u)=c(v)$.

Suppose that $D$ is a dominating set of $G$ with  $C= c(D\cap V_G(T))$ that has the minimum size. We claim that $D\subseteq V(G')$. To see it assume that there is $u\in D\setminus V(G')$. Then $u$ is an $e$-vertex of $G$ for some $e=xy\in A$. Then there is $v\in V_G(T)$ such that $x,y\in M_v$ and $c(v)\in C$. Since $D$ has a vertex $v'$ with $c(v')=c(v)$, we have that $x,y\in M_{v'}$. Clearly, $M_u\subseteq M_{v'}$. This implies that $D\setminus\{u\}$ is a dominating set of $G$. Since $u$ is an $e$-vertex, $u$ is not colored and the obtained dominating set contains vertices with all the colors from $C$, but this contradicts the choice of $D$. Therefore, $D\subseteq V(G')$. Because $G[V(G')]$ is a subgraph of $G'$, we have that $D$ is a dominating set of $G'$.

Let now $D$ be a dominating set of $G'$ with $C= c(D\cap V_G(T))$ that has the minimum size. We show that $D$ is a dominating set of $G$. 

First, we prove that $D$ dominates every  $u\in V(G)\setminus V(G')$. If $u\in V(G)\setminus V(G')$, then $u$ is an $e$-vertex for some $e=xy\in A$. We have that there is $v\in V_G(T)$ such that $x,y\in M_v$ and $c(v)\in C$. Because $C\subseteq c(D\cap V_{G}(T))$, there is $v'\in D$ with $c(v')=c(v)$. This implies that $x,y\in M_{v'}$ in $G$ and, therefore, $v'$ dominates $u$.

Now we show that $D$ dominates the vertices of $V(G')$ in $G$. To obtain a contradiction, assume that there is $u\in V(G')$ that is not dominated by $D$ in $G$. As $D$ is a dominating set of $G'$, there is $v\in D$ such that $uv\in E(G')$ and $uv\notin E(G)$. It follows that there are $x,y\in V(T)$ such that $x\in M_u$, $y\in M_v$ and for the $(x,y)$-path $P$ in $T$, $E(P)\subseteq A$.  Let $xy'$ be the edge of $P$ that is incident to $x$. Since $xy'\in A$, there is $w\in V_G(T)$ such that $x,y\in M_w$  in $G$ and $c(w)\in C$. Because $C\subseteq c(D\cap V_{G}(T))$, there is $w'\in D$ with $c(w')=c(w)$. This implies that $x\in M_{w'}$ in $G$ and, therefore, $w'$ dominates $u$.
\qed
\end{proof}

\begin{figure}[ht]
\centering
\scalebox{0.75}{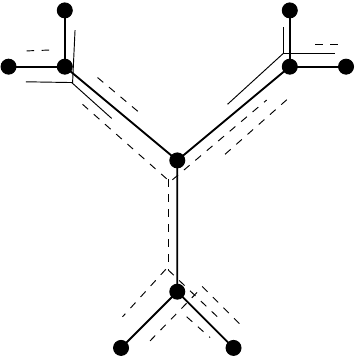}
\caption{ $A'=\{xy\}$; the models of vertices of $V_G(T)$ with colors from $C$ are shown by solid lines and the other models are shown by dashed lines.   
\label{fig:lem:contr-uncol}}
\end{figure}

\begin{lemma}\label{lem:contraction-uncolored}
Let $C\subseteq c(V_G(T))$ and  
\begin{align*}
A'=\{xy\in E(T)\mid& x,y\in M_u\text{ for some~}u\in V_G(T)\text{~s.t.~}c(u)\notin C\text{ and }\\
&x,y\notin M_v\text{ for all~}v\in V_G(T)\text{~s.t.~} c(v)\in C\},
\end{align*}
and assume that for every $e\in A'$, $G$ has no $e$-vertices (see \autoref{fig:lem:contr-uncol} for an example). 
Further, let $G'$ be the graph obtained from $G$ by iteratively contracting of edges of $A'$ in $T$.  Then for any set $D\subseteq V(G)$, $D$ is a minimum dominating set of $G$ satisfying the condition $C= c(D\cap V_G(T))$ if and only if $D$ is a dominating set of $G'$ satisfying the same condition  $C= c(D\cap V_{G}(T))$.
\end{lemma} 

\begin{proof}
Observe that $V(G')=V(G)$, because we do not delete $e$-vertices when we contract $A'$. Notice also that the model $M_u$ of a vertex $u\in V_G(T)$ with $c(u)\in C$ remains the same in the modified representation obtained by the contraction. 

If $D$ is a dominating set of  $G$ with  $C=c(D\cap V_G(T))$, then it is straightforward  to verify that $D$ is a dominating set of $G'$. Suppose that 
$D$ is a dominating set of  $G'$ with  $C=c(D\cap V_G(T))$. We have that $G$ is a subgraph of $G'$. Suppose that $u,v\in V(G)$ are adjacent in $G'$ but are not adjacent in $G$. It follows that there are $x,y\in V(T)$ such that $x\in M_u$, $y\in M_v$ and there is an $(x,y)$-path $P$ in $T$ such that $E(P)\subseteq A'$.
Then $c(u),c(v)\notin C$ by the definition of $A'$. 
Hence, $u,v\notin D$. We obtain that for every $u,v\in V(G)$ that are adjacent in $G'$ but are not adjacent in $G$, $u,v\notin D$. 
Hence, $D$ is a dominating set of $G$.%
\qed
\end{proof}

We say that $\mathcal{M}$ is a \emph{nice} representation if $|M_v\cap V(T)|\leq 1$ for each $v\in V(G)$, i.e., each set $M_v$ contains at most one branching node of $T'$.

We say that a $T$-representation $\mathcal{M}=\{M_v\}_{v\in V(G)}$ of $G$ with respect to a subdivision  $T'$ is an \emph{$r$-rooted} representation if a node $r$ of $T$ is chosen to be a root. The root defines the parent-child relation on $V(T)$ and $V(T')$. For $x\in V(T)$ ($x\in V(T')$), we denote by $T_x$ ($T_x'$ respectively) the subtree of $T$ ($T'$ respectively) induced by $x$ and its descendants.  For $x\in V(T)$ and any child $y$ of $x$, we denote by $T_{xy}$  the subtree of $T$ induced by $x$, $y$ and the descendants of $y$.
For each $x \in V(T)$, we use $V_x(G)$ to denote the vertices of whose models
use nodes of the subtree rooted at $x$, that is, \[V_x(G)=\{v\in V(G)\mid M_v\cap V(T_x')\neq\emptyset\}.\]
For a child $y$ of $x$, we denote by $V_{xy}(G)$ the union of the set of vertices of $V_y(G)$ and the $xy$-vertices, that is,
\[V_{xy}(G)=\{v\in V(G)\mid x\notin M_v\text{ and }M_v\cap V(T_{xy}')\neq\emptyset\}.\]

Let $\mathcal{M}=\{M_v\}_{v\in V(G)}$ be a nice $r$-rooted $T$-representation  of $G$ with respect to $T'$. For any vertex $x\in V(T')$, we denote by $\dist(x)$ the distance between $r$ and $x$ in $T'$. For every set $X\subseteq V(T')$, we define
\[
  d_{\min}(X)=\min\{\dist(x)\mid x\in X\} \quad \text{and} \quad d_{\max}(X)=\max\{\dist(x)\mid x\in X\}.
\]
Consider an edge $e=xy\in E(T)$, where $y$ is a child of $x$, and denote by $P_e$ the $(x,y)$-path in $T'$ corresponding to $e$. 
Consider $u,v\in V(G)$ such that $M_u^e=M_u\cap V(P_e)\neq\emptyset$ and $M_v^e=M_v\cap V(P_e)\neq\emptyset$. We write $u\preceq_e v$ if either
\begin{itemize}
\item $d_{\min}(M_u^e)<d_{\min}(M_v^e)$; or
\item $d_{\min}(M_u^e)=d_{\min}(M_v^e)$ and $d_{\max}(M_u^e)\leq d_{\max}(M_v^e)$.
\end{itemize}
Respectively, $u\prec_e v$ if either
\begin{itemize}
\item  $d_{\min}(M_u^e)<d_{\min}(M_v^e)$; or 
\item $d_{\min}(M_u^e)=d_{\min}(M_v^e)$ and $d_{\max}(M_u^e)< d_{\max}(M_v^e)$.
\end{itemize}
%We extend this notation and also write $u\preceq_e v$ (or $u\prec_e v$) if $y\notin M_v\subseteq V(T_y')$. For $u\in V(G)$ with $M_u\cap V(P_e)$, $V(\succeq_e(u))=\{v\in V(G)\mid u\preceq_e v\}$.

We consider the following auxiliary problem for $T$-graphs with nice representations.

\begin{center}
  \fbox{\begin{minipage}{0.95\textwidth}
      \noindent\textsc{Dominating Set Extension}
        \begin{description}
        \item[Input:] A tree $T$  and a  graph $G$ with a  
%nice $r$-rooted 
$T$-representation  of $G$,
% with respect to a subdivision $T'$ of $T$, 
 positive integers $k$ and $d$,
          a \emph{labeling} function $c\colon \bigcup_{x\in V(T)}V_G(x)\rightarrow \mathbb{N}$,   and a collection of sets $\{C_x\}_{x\in V(T)}$ of size at most $d$ where for each $C_x\subseteq c(V_G(x))$   (some sets could be empty) such that for  every dominating set $D$ of $G$ of minimum size with the properties that  
          \begin{itemize}   
          \item[(a)]  $D$ has  at most $d$ $x$-vertices for each $x\in V(T)$,
          \item[(b)] for each $x\in V(T)$, $C_x\subseteq c(D\cap V_G(x))$,
          \end{itemize} 
          it holds that the number of nodes $x\in V(T)$ such that $D$ contains an $x$-vertex is maximum and  for each $x\in V(T)$, $C_x= c(D\cap V_G(x))$.
        \item[Task:] Decide whether there is a dominating set $D'$ of $G$ of size at most $k$ containing at most  $d$ $x$-vertices for $x\in V(T)$ such that
          for each $x\in V(T)$,  $C_x=c(D'\cap V_G(x))$.
        \end{description}
      \end{minipage}}
  \end{center}
  
Note that \textsc{Dominating Set Extension} is a \emph{promise} problem: we are promised that there is $D$ with the described properties but $D$ itself is not given.
Moreover,  the promise could be false but we are not asked  to verify it.

\begin{lemma}\label{lem:promise}
Given a nice $r$-rooted representation $T$ of the input graph where $T$ is a tree with at most $\ell$ leaves,
  \textsc{Dominating Set Extension} can be solved in time $2^{\Oh((d+\ell)\log d)}\cdot n^{\Oh(1)}$.  Moreover, it can be done by an algorithm that either returns a correct yes-answer or (possibly incorrect) no-answer even if the promise is false.
%$2^{\Oh(\ell\log\ell)}n^{\Oh(1)}$. 
%$2^{\Oh(\ell^2)}n^{\Oh(1)}$. 
%where $\ell$ is the number of leaves of $T$.
\end{lemma}

\begin{proof}
Let $G$ be a $T$-graph for a tree $T$ rooted in $r$ that has a nice $r$-rooted $T$-representation $\mathcal{M}$ with respect to a subdivision $T'$ of $T$. 
If $G$ is disconnected, then we reduce the problem to solving \textsc{Dominating Set Extension} for the components of $G$. Assume from now that $G$ is connected.
Let also $k$ and $d$ be positive integers,  $c\colon \bigcup_{x\in V(T)}V_G(x)\rightarrow \mathbb{N}$ be a labeling function,   and $\{C_x\}_{x\in V(T)}$ be  a collection of sets  where each $C_x\subseteq c(V_G(x))$.
%Let also  $\{d_x\}_{x\in V(T)}$ be a collection of nonnegative integers satisfying the promise of \textsc{Dominating Set Extension}. 
Since $G$ is connected, we can assume without loss of generality that for every $x\in V(T')$, there is $v\in V(G)$ with $x\in M_v$. Otherwise, we can replace $T$ by its subtree or reduce the number 
of subdivision nodes of $T'$
 without increasing the number of leaves. We also assume that $|C_x|\leq d\leq |V_G(x)|$ as, otherwise, we have a trivial no-instance of \textsc{Dominating Set Extension}.
We construct a dynamic programming algorithm for the problem that finds the minimum size of a  dominating set $D$ of $G$  containing at most  $d$ $x$-vertices for $x\in V(T)$ such that
 for each $x\in V(T)$,  $C_x=c(D\cap V_G(x))$. Our algorithm assumes that the promise is fulfilled for the considered instance of \textsc{Dominating Set Extension}.
The algorithm uses the properties that every node of $T$ has at most $\ell$ children, i.e., the number of the children is bounded by the parameter, 
and for each edge $e$ of $T$, the set of $e$-vertices of $G$ composes an interval graphs for which the domination problem can be solved efficiently.

First, we construct a subroutine that solves the following auxiliary problem for each $e=xy\in E(T)$. Let $P_e$ be the $(x,y)$-path in $T'$ corresponding to $e$.
Let $U_e=\{v\in V(G)\mid M_v\cap V(P_e)\neq\emptyset\}$. For $X\subseteq U_e$,  $\alpha_e(X)$ is the minimum size of a set $S$ of $e$-vertices of $G$ that dominates $X$; we assume that $\alpha_e(X)=0$ if $X=\emptyset$ and $\alpha_e(X)=+\infty$ if such a dominating set of $e$-vertices $S$ does not exist.

\begin{claimn}\label{claimA}
  For every $e\in E(T)$ and  $X\subseteq U_e$, $\alpha_e(X)$ can be computed in time $n^{\Oh(1)}$.
\end{claimn}

\begin{proof}%[Proof of \autoref{claimA}]
If $X=\emptyset$, then $\alpha_e(X)=0$ by the definition. Assume that $X\neq\emptyset$.
If there is a vertex in $X$ that is not dominated by any $e$-vertex of $G$, then we set $\alpha_e(X)=+\infty$. Otherwise, it is straightforward to see that $S$ exists, and we construct $S$ using the well-known greedy approach for constructing a minimum dominating set in an interval graph. 

Initially, we set $S=\emptyset$ and then increase it iteratively until all the vertices of $X$ are dominated. Denote by $Y\subseteq X$ the set of vertices that are not dominated by the the current $S$. Then we do the following:
\begin{enumerate}
\item Find a vertex $w$ in $Y$ that is maximum with respect to the ordering $\preceq_e$. 
\item Find a minimum with respect to $\preceq_e$ $e$-vertex $v$ that dominates $w$, set $S=S\cup \{v\}$ and recompute $Y$.
\item If $Y\neq\emptyset$, then return to Step~1.  
\end{enumerate}

It is straightforward to verify that the algorithm correctly computes $\alpha_e(X)$ in polynomial time. This proves \autoref{claimA}.%
\cqed
\end{proof}

We say that a node $x\in V(T)$ is \emph{loaded} if $C_x\neq\emptyset$ and $x$ is \emph{unloaded} otherwise.
We say that a set of vertices $S$ is \emph{extendable} if there is a dominating set $D$  such that 
\begin{enumerate}[(a)]
\item $D$ has  at most $d$ $x$-vertices for $x\in V(T)$,
\item for each $x\in V(T)$, $C_x\subseteq c(D\cap V_G(x))$,
\item $D$ has the minimum size, contains $S$, and the conditions of the promise are fulfilled:
 the number of nodes $x\in V(T)$ such that $D$ contains an $x$-vertex is maximum and  for each $x\in V(T)$, $C_x= c(D\cap V_G(x))$.
\end{enumerate} 

We are ready to explain our dynamic programming algorithm for \textsc{Dominating Set Extension}. It works on $T$ starting from the leaves and moving towards the root.
To avoid dealing with the root that has no parent separately, we add an artificial node $r'$ to $T$ and $T'$ and make $r'$ the parent of $r$ and the new root. Observe that $r'$ is the unique node of the tree that is not included in $M_v$ for any $v\in V(G)$. 

We start with defining the tables of data that the algorithm stores for each $x\neq r'$ of $T$. Let $y$ be a parent of $x$ and let $e=yx$. Consider the set 
$W=\{v_1,\ldots,v_p\}$ of $x$-vertices of $G$ and assume that $v_1\preceq_e\ldots\preceq_e v_p$ (note that if $x=r$ and $y=r'$, then the ordering is arbitrary). 
The tables constructed in different ways for loaded and unloaded vertices, because if $x$ is loaded, then some vertices whose models contain $x$ should be included in every partial solutions and,  
otherwise, the partial solutions should exclude such vertices.     

If $x$ is loaded, then for $x$ and $i\in \{1,\ldots,p\}$, the algorithm stores the value $\beta(x,i)$ that  is either the minimum size of a set $S\subseteq V_x(G)$ such that 
\begin{enumerate}[(i)]
\item \label{beta1} $v_i\in S$,
\item $S$ contains at most $d$ $z$-vertices for each $z\in V(T_x)$,
\item for each $z\in V(T_x)$, $C_z=c(S\cap V_G(z))$, and
\item \label{beta4} $S$ dominates all the vertices of $V_x(G)$,
\end{enumerate}
or we may set $\beta(x,i)=+\infty$ if we detect that there is no $S\subseteq V_x(G)$ satisfying \eqref{beta1}--\eqref{beta4} such that $S$ is extendable. In particular, it can happen if every set $S$ satisfying \eqref{beta1}--\eqref{beta4} contradicts the promise of \textsc{Dominating Set Extension}.

Similarly,
if $x$ is unloaded, then for $x$ and $i\in\{0,\ldots,p\}$ the algorithm stores the value $\gamma(x,i)$ that  is either the minimum size of a set $S\subseteq V_x(G)$ such that 
\begin{enumerate}[(i)]
  \setcounter{enumi}{4}
\item \label{gamma5} $v_{i+1},\ldots,v_{p}$ are dominated by $S$,
\item \label{gamma6} $S$ contains at most $d$ $z$-vertices for each $z\in V(T_x)$, 
\item for each $z\in V(T_x)$, $C_z=c(S\cap V_G(z))$, and
\item \label{gamma8} $S$ dominates all the vertices of $V_x(G)\setminus \{v_1,\ldots,v_p\}$,
\end{enumerate}
or  we may set $\gamma(x,i)=+\infty$ if we detect that there is no $S\subseteq V_x(G)$ satisfying \eqref{gamma5}--\eqref{gamma8} such that $S$ is extendable. Similarly to $\beta(x,i)$, we do it if we detect that every set $S$ satisfying \eqref{gamma5}--\eqref{gamma8} contradicts the promise of \textsc{Dominating Set Extension}. In particular, it happens when we gain by including an $x$-vertex of $G$ into a partial solution, that is, we can increase the number of $x$-vertices for $x\in V(T)$ in the partial solution without decreasing its size. 

It also is assumed that $\beta(x,i)=+\infty$ and $\gamma(x,i)=+\infty$ if there is no $S$ that satisfies the conditions \eqref{beta1}--\eqref{beta4} or \eqref{gamma5}--\eqref{gamma8} respectively. 

Now we explain how we compute the values $\beta(x,i)$ and  $\gamma(x,i)$. First, we do it for leaves. 

\paragraph{\bf Computing $\beta(x,i)$ and  $\gamma(x,i)$ for leaves.}When $x$ is a loaded leaf of $T$, we set:
\[
  \beta(x,i)=
  \begin{cases}
    |C_x|&\mbox{ if } c(v_i)\in C_x,\\
    +\infty&\mbox{ otherwise}.
  \end{cases}
\]
%It is straightforward to see that $\beta(x,i)$ is computed correctly.

\noindent When $x$ is an unloaded leaf, we set:
%Since $V_x(G)=\{v_1,\ldots,x_p\}$ and $C_x=\emptyset$, we immediately obtain that 
\[
  \gamma(x,i)=
  \begin{cases}
    0&\mbox{ if }i=p,\\
    +\infty&\mbox{ if }i<p.
  \end{cases}
\]

%\medskip
Now we compute the values $\beta(x,i)$ and  $\gamma(x,i)$  for non-leaves.
Suppose that $x\neq r'$ is a non-leaf  node of $T$. Let $z_1,\ldots,z_s$ be the children of $x$ in $T$ and let $e_j=xz_j$ for $j\in\{1,\ldots,s\}$. Assume that the functions $\beta$ and $\gamma$  are computed for the children of $x$ depending on whether they are loaded or not.

\paragraph{\bf Computing $\beta(x,i)$ for loaded $x$.} 
%We assume that $i\in\{1,\ldots,p\}$ if fixed and explain how $\beta(x,v_i)$ is computed.

Consider $z_j$ for $j\in\{1,\ldots,s\}$. Let $W_j=\{u_1,\ldots,u_q\}$ be the set of $z_j$-vertices of $G$. We assume that $u_1\preceq_{e_j}\ldots\preceq_{e_j}u_q$.
As a first step we compute, for every $i'\in\{1,\ldots,p\}$, the value of the auxiliary function $\delta_j(x,i')$ defined as follows, depending whether $z_j$ is loaded or not.
Informally, $\delta_j(x,i')$ is the minimum size of a set of vertices $S_j$ (satisfying the constraints for considered partial solutions) of  $V_{xz_j}(G)$ that dominate the vertices of $V_{xz_j}(G)$ which are not dominated by $v_{i'}$. In particular, this means that if we add $v_{i'}$ to a partial solution, we dominate $V_{xz_j}(G)$.

\paragraph{Case~1.} The vertex $z_j$ is loaded.  For each $i'\in \{1,\ldots,p\}$ and $h\in\{1,\ldots,q\}$, let 
$$X_{i',h}=\{w\in V(G)\mid w\in V_G(e_j),~M_w\cap(M_{v_{i'}}\cup M_{u_h})=\emptyset\}. $$
Recall that $\alpha_{e_j}(X)$ is the minimum size of a set $S$ of $e_j$-vertices of $G$ that dominates $X$. 
We set
$$\delta_j(x,i')=\min_{1\leq h\leq q}\{\beta(z_j,h)+\alpha_{e_j}(X_{i',h})\}.$$

\paragraph{Case~2.} The vertex $z_j$ is unloaded.  For $i'\in\{1,\ldots,p\}$, let 
$$Y_{i'}=\{w\in V(G)\mid w\in V_G(e_j),~M_w\cap M_{v_{i'}}=\emptyset\}$$
and for $h\in\{0,\ldots,q\}$,
$$
Y_{i',h}=
\begin{cases}
Y_{i'}&\mbox{ if }h=0,\\
Y_{i'}\cup\{u_h\}&\mbox{ if }h>0.
\end{cases}
$$
We set 
$$\delta_j(x,i')=\min_{0\leq h\leq q}\{\gamma(z_j,h)+\alpha_{e_j}(Y_{i',h})\}.$$

\medskip
We have that $\delta_j(x,i')$ is defined in both cases. Note that if 
 $v_{i'}$ is included in a partial solution, then it can be used to dominate  vertices of $V_{xz_j}(G)$  for distinct $j\in\{1,\ldots,s\}$.
Therefore, we extend the definition of this function on subsets of $\{1,\ldots,s\}$.
For $i'\in\{1,\ldots,p\}$ and 
each non-empty $J\subseteq\{1,\ldots,s\}$, we set
\[
  \delta_{J}(x,i')=\sum_{j\in J}\delta_j(x,i')\quad \text{and} \quad \delta_{\emptyset}(x,i')=0.
\]

We consider all possible partitions $\mathcal{P}=\{J_1,\ldots, J_t\}$ of $\{1,\ldots,s\}$
for all $t
\in \{|C_x|, \dots, d\}$, where  
some 
sets may be empty.
%$0\in J_1$ and some other sets could be empty. 
We then consider all possible surjections $\varphi\colon \{1,\ldots,t\}\rightarrow C_x$. 
%with $c_x(v_i)=\varphi(1)$.
For every $x$ and
$i\in\{1,\ldots,p\}$,
%\begin{equation}\label{eq:beta}
%\beta(x,i)=\min\{t+\delta_{J_1}(x,i)\eta(\varphi,i)+\sum_{h=2}^t\min\{\delta_{J_h}(x,i')\mid 1\leq i'\leq p, c(v_{i'})=\varphi(h)\}\mid  \mathcal{P},~\varphi\},
%\end{equation} 
%where 
%$
%\eta(\varphi,i)=
%\begin{cases}
%1& \mbox{ if }\varphi(1)=c(v_i),\\
%+\infty& \mbox{otherwise}.
%\end{cases}
%$
we set
\begin{equation}\label{eq:beta}
\beta(x,i)=\min \left \{t+\sum_{h=1}^t\min \left \{\delta_{J_h}(x,i')\mid 1\leq i'\leq p, i'=i\text{ if }h=1, c(v_{i'})=\varphi(h) \right \}\mid  \mathcal{P},~\varphi \right \},
\end{equation} 
where the minimum is taken over all $\mathcal{P}$ and $\phi$;
 we assume to simplify notation that 
\[
  \min\{\delta_{J_h}(x,i')\mid 1\leq i'\leq p,  i'=i\text{ if }h=1, c(v_{i'})=\varphi(h)\}=+\infty
\] if 
$h\geq 2$ and 
there is no $i' \in \{1, \dots, p\}$ such that $c(v_{i'})=\varphi(h)$ or if $h=1$ and $\varphi(1)\neq c(v_i)$.
%if $c(v_i)=\varphi(1)$, and $\beta(x,i)=+\infty$ otherwise.
The intuition behind (\ref{eq:beta}) is the following. We are selecting $t$ $x$-vertices $v_{i_1},\ldots,v_{i_t}$ with $i=i_1$ (recall that $v_i$ should be in a solution) to include them in a partial solution. The partition $\{J_1,\ldots,J_t\}$ encodes the property that each $v_{i_h}$ is used to dominate some vertices of $V_{xz_j}(G)$  for $j\in J_h$. The function $\varphi$ encodes colors of the selected vertices.

\paragraph{\bf Computing $\gamma(x,i)$ for unloaded $x$.} 
Consider $z_j$ for each $j\in\{1,\ldots,s\}$. Let $W_j=\{u_1,\ldots,u_q\}$ be the set of $z_j$-vertices of $G$. 
We assume that $u_1\preceq_{e_j}\ldots\preceq_{e_j}u_q$. Recall that $W$ denotes the set of $x$-vertices.
Consider also the ordering $v_{i_1^j},\ldots,v_{i_p^j}$ of the vertices of $W$ such that $v_{i_1^j}\preceq_{e_j}\ldots\preceq_{e_j}v_{i_p^j}$.
We use the following crucial property to compute $\gamma(x,i)$. Since $x$ is unloaded, the vertices of $W$ should be dominated by vertices whose models do not contain $x$ and for each $j\in\{1,\ldots,s\}$, the vertices of $V_G(xz_j)$  should be dominated by some vertices from this set. Let $S_j\subseteq V_{xz_j}(G)$ is the set of vertices that dominate 
 the vertices of $V_G(xz_j)$ and some vertices of $W$. If $S_j$ is not a set of minimum size (satisfying the constraints of \textsc{Dominating Set Extension}) dominating the vertices 
 of $V_G(xz_j)$, then $S_j$ can be replaces by a set $S_j'$ of minimum size and a vertex of $W$ that is used to dominate the vertices of $W$. This way, we obtain a partial solution that contains an $x$-vertex contradicting the promise of \textsc{Dominating Set Extension}. This means that to dominate the  vertices of $V_G(xz_j)$, we should select a set $S_j\subseteq V_{xz_j}(G)$ of minimum size that contains a vertex whose model is at minimum possible distance from $x$.
To exploit this property, we define the auxiliary functions $\eta(j)$ and $\psi(j)$ as follows depending on whether $z_j$ is loaded or not. The value of $\eta(j)$ is the minimum size of a set of vertices $S_j$ that dominate  the  vertices of $V_G(xz_j)$. The value of $\psi(j)$ is the minimum $t\in\{1,\ldots,p+1\}$ such that there is a set $S_j$ of size $\eta(j)$ that additionally dominates the vertices $v_{i_t^j},\ldots,v_{i_p^j}$ (if $t=p+1$, then the vertices of $W$ are not dominated by $S_j$).

\paragraph{Case~1.} The vertex $z_j$ is loaded. For each $h\in \{1,\ldots,q\}$, let 
$$X_{h}=\{w\in V(G)\mid  w\in V_G(e_j),~M_w\cap M_{u_h}=\emptyset\},$$
and let 
\begin{equation}\label{eq:eta-a}
\eta(j)=\min_{1\leq h\leq q}( \beta(z_j,h)+\alpha_{e_j}(X_h)).
\end{equation}
For each $t\in\{1,\ldots,p+1\}$ and $h\in \{1,\ldots,q+1\}$, denote
$$X_{h,t}=
\begin{cases}
X_h\cup\{v_{i_t^j},\ldots,v_{i_p^j}\} &\mbox{ if }t\leq p,\\
X_h&\mbox{ if }t=p+1,  
\end{cases}
$$
and let 
\begin{equation}\label{eq:psi-a}
\psi(j)=\min\{t\mid 1\leq t\leq p+1,\eta(j)=\min_{1\leq h\leq q}( \beta(z_j,h)+\alpha_{e_j}(X_{h,t}))\}.
\end{equation}

\paragraph{Case~2.} The vertex $z_j$ is unloaded. For each $h\in \{0,\ldots,q\}$, let
$$
Y_h=
\begin{cases}
\{w\in V(G)\mid   w\in V_G(e_j)\}&\mbox{ if }h=0,\\
\{w\in V(G)\mid   w\in V_G(e_j)\}\cup\{u_1,\ldots,u_h\}&\mbox{ if }h\geq 1,
\end{cases}
$$
and let 
\begin{equation}\label{eq:eta-b}
\eta(j)=\min_{0\leq h\leq q} (\gamma(z_j,h)+\alpha_{e_j}(Y_h)).
\end{equation}
For each $t\in\{1,\ldots,p+1\}$ and $h\in \{1,\ldots,q+1\}$, denote
$$Y_{h,t}=
\begin{cases}
Y_h\cup\{v_{i_t^j},\ldots,v_{i_p^j}\} &\mbox{ if }t\leq p,\\
Y_h&\mbox{ if }t=p+1,  
\end{cases}
$$
and let 
\begin{equation}\label{eq:psi-b}
\psi(j)=\min\{t\mid 1\leq t\leq p+1,\eta(j)=\min_{0\leq h\leq q}( \gamma(z_j,h)+\alpha_{e_j}(Y_{h,t}))\}.
\end{equation}

\medskip
Now for each $i\in\{0,\ldots,p\}$, we set
\begin{equation}\label{eq:gamma}
\gamma(x,i)=
\begin{cases}
\sum_{j=1}^s\eta(j)&\mbox{ if }\text{for }h\in\{i+1,\ldots,p\},~\text{there is }j\in\{1,\ldots,s\}\text{ s.t. }v_{i_{\psi(j)}^j}\preceq_{e_j} v_h\\
+\infty&\mbox{ otherwise}.  
\end{cases}
\end{equation}

%%%%%%%%%%%%%%%%%%%%%%%
\medskip
We compute $\beta$ and $\gamma$ for all nodes of $T$ except the artificial root $r'$. The algorithm is based on the following properties of these values.

\begin{claimn}\label{claimB}
If $x$ is a loaded node and $i\in \{1,\ldots,p\}$, then the following is fulfilled:
  \begin{itemize}
  \item if $\beta(x,i)<+\infty$, then there is a set $S\subseteq V_x(G)$ satisfying the conditions \eqref{beta1}--\eqref{beta4} of size at most $\beta(x,i)$; and
\item if there is a  set $S\subseteq V_x(G)$ of minimum size satisfying the conditions \eqref{beta1}--\eqref{beta4} that is extendable, then $|S|=\beta(x,i)$.
  \end{itemize}
  Similarly,  if $x$ is unloaded, then for $x$ and $i\in \{0,\ldots,p\}$, the following is fulfilled:
  \begin{itemize}
  \item if $\gamma(x,i)<+\infty$, then 
 there is a set $S\subseteq V_x(G)$ satisfying the conditions \eqref{gamma5}--\eqref{gamma8} of size at most $\gamma(x,i)$; and
\item if there is a  set $S\subseteq V_x(G)$ of minimum size 
satisfying the conditions \eqref{gamma5}--\eqref{gamma8} that is extendable, then  $|S|=|\gamma(x,i)$.
  \end{itemize}
\end{claimn}

\begin{proof}%[Proof of \autoref{claimB}]
It is straightforward to verify these properties for the leaves of $T$ by the definition of $\beta$ and $\gamma$. 

We use standard approach for proving correctness of dynamic programming algorithms. We assume inductively that the properties of the  values of $\beta$ and $\gamma$ are fulfilled  for the children of $x$ using as the  base of the induction the fact that we already verified the properties of  $\beta$ and $\gamma$ for the leaves of $T$.

As in the description of the algorithm, we assume that $W=\{v_1,\ldots,v_p\}$ is the set of $x$-vertices and assume that $v_1\preceq_e\ldots\preceq_e v_p$ where $e=yx$ and $y$ is the parent of $x$. In the same way, $z_1,\ldots,z_s$ are the children of $x$. 

First, we prove the claim for $\beta$. Let $x$ be a loaded node and let $i\in\{1,\ldots,p\}$. 

%
%
%
%Let 
% \begin{equation}\label{eq:beta-a}
%\beta^*(x,i)=\min\{t+\sum_{h=1}^t\min\{\delta_{J_h}(x,i')\mid 1\leq i'\leq p, i'=i\text{ if }h=1, c(v_{i'})=\varphi(h)\}\mid  \mathcal{P},~\varphi\}.
%\end{equation}
%We show first that $\beta(x,i)\leq \beta^*(x,i)$ and then prove the opposite inequality.

Let $\beta(x,i)<+\infty$. We show that there is a set $S\subseteq V_x(G)$ of size at most $\beta(x,i)$  such that 
\begin{enumerate}[(i)]
\item $v_i\in S$, \label{it:rc1}
\item $S$ contains at most $d$ $z$-vertices for each $z\in V(T_x)$,
\item for each $z\in V(T_x)$, $C_z=c(S\cap V_G(z))$, and \label{it:rc3}
\item $S$ dominates all the vertices of $V_x(G)$. \label{it:rc4}
\end{enumerate}

 Consider a partition  
$\mathcal{P}=\{J_1,\ldots, J_t\}$ of $\{1,\ldots,s\}$ and a mapping $\varphi\colon \{1,\ldots,t\}\rightarrow C_x$ for which the minimum in the right part of~(\ref{eq:beta}) is achieved. Further, for each $j\in\{1,\ldots,s\}$, let $i_j$ be a value of $i'$ for which 
$\min\{\delta_{J_h}(x,i')\mid 1\leq i'\leq p, i'=i\text{ if }h=1, c(v_{i'})=\varphi(h)$ is achieved; note that $i_1=i$. 

We define $S_x=\{v_{i_1},\ldots,v_{i_t}\}$; observe that some vertices could be repeated and in this case we remove the duplicates. 
We have that  $|S_x|\leq t\leq d$ and $v_i\in S_x$.  Since $\varphi$ is a surjection,  
$C_x=c(S_x\cap V_G(x))$. We obtain that \eqref{it:rc1}--\eqref{it:rc3} are fulfilled for $z=x$.
Clearly, all $x$-vertices of $G$ are dominated by $S_x$, so \eqref{it:rc4} also holds.

Consider $z_j$ for $j\in\{1,\ldots,s\}$. As in the description of the algorithm, we assume that  $W_j=\{u_1,\ldots,u_q\}$ is the set of $z_j$-vertices of $G$. We also assume that $u_1\preceq_{e_j}\ldots\preceq_{e_j}u_q$. Let $j\in J_{i'}$. We consider two cases depending on whether $z_j$ is loaded or not.

\paragraph{Case~1.} The vertex $z_j$ is loaded. Let $h\in\{1,\ldots,q\}$ be such that $\beta(z_j,h)+\alpha_{e_j}(X_{i',h})$ has the minimum value. 
By the inductive assumption,
 there is $S_{z_j}\subseteq V_{z_j}(G)$ of size at most $\beta(z_j,h)$  such that 
\begin{enumerate}[(i')]
\item $u_h\in S_j$,
\item $S_{z_j}$ contains at most $d$ $z$-vertices for each $z\in V(T_{z_j})$,
\item for each $z\in V(T_{z_j})$, $C_z=c(S_{z_j}\cap V_G(z))$, and
\item $S_{z_j}$ dominates all the vertices of $V_{z_j}(G)$.
\end{enumerate}

Consider $X_{i',h}=\{w\in V(G)\mid w\in V_G(e_j),~M_w\cap(M_{v_{i'}}\cup M_{u_h})=\emptyset\}$. Notice that if $w$ is an $e_j$-vertex and $w\notin X_{i',h}$, then $w$ is dominated either by $v_{i'}\in S_x$ or $u_h\in S_{z_j}$, i.e., $w$ is dominated by $S_x\cup S_{z_j}$.  By the definition of $\alpha_{e_j}(X_{i',h})$, there is a set $S_{xz_j}$ of $e_j$-vertices of size $\alpha_{e_j}(X_{i',h})$ that dominates $X_{i',h}$.

\paragraph{Case~2.} The vertex $z_j$ is unloaded. Let $h\in\{1,\ldots,q\}$ be such that $\gamma(z_j,h)+\alpha_{e_j}(Y_{i',h})$ has the minimum value. 

By the inductive assumption, there is $S_{z_j}\subseteq V_{z_j}(G)$ of size at most $\gamma(z_j,h)$  such that 
\begin{enumerate}[(i')]
  \setcounter{enumi}{4}
\item $u_{h+1},\ldots,u_{q}$ are dominated by $S_{z_j}$,
\item $S_{z_j}$ contains at most $d$ $z$-vertices for each $z\in V(T_{z_j})$, 
\item for each $z\in V(T_{z_j})$, $C_z=c(S\cap V_G(z))$, and
\item $S_{z_j}$ dominates all the vertices of $V_{z_j}(G)\setminus \{u_1,\ldots,u_q\}$.
\end{enumerate}

Recall that we defined
\begin{align*}
  Y_{i'} & = \{w\in V(G)\mid w\in V_G(e_j),~M_w\cap M_{v_{i'}}=\emptyset\},\ \text{and}\\
  Y_{i',h} & =
  \begin{cases}
    Y_{i'}&\mbox{ if }h=0,\\
    Y_{i'}\cup\{u_h\}&\mbox{ if }h>0.
  \end{cases}
\end{align*}
If $w$ is a $e_j$-vertex and $w\notin X_{i',h}$, then $w$ is dominated by $v_{i'}\in S_x$, i.e., $w$ is dominated by $S_x\cup S_{z_j}$. 
 By the definition of $\alpha_{e_j}(Y_{i',h})$, there is a set $S_{xz_j}$ of $e_j$-vertices of size $\alpha_{e_j}(Y_{i',h})$ that dominates $Y_{i',h}$. This set dominates the 
$e_j$-vertices that are not dominated by $S_x$ and the $z_j$-vertices $u_1,\ldots,u_h$ if $h\geq 1$ that are the only vertices of $V_{z_j}(G)$ that (possibly) are not dominated by $S_{z_j}$.

Now we let 
\[
  S=S_x\cup\left (\bigcup_{j=1}^s\left (S_{z_j}\cup S_{xz_j}\right )\right ).
  \]
We have that \eqref{beta1}--\eqref{beta4} are fulfilled for $S$. It remains to notice that $|S|\leq \beta(x,i)$ by the definition.

Assume now that $S\subseteq V_x(G)$ is a set of minimum size satisfying (i)--(iv) and $S$ is extendable. Since $S$ is a set of minimum size satisfying these conditions,
 as we already proved, $|S|\leq\beta(x,i)$. We prove that $|S|\geq \beta(x,i)$.

%\medskip
%Now, we prove that $\beta(x,i)\geq \beta^*(x,i)$. If $\beta(x,i)=+\infty$, then the inequality is trivial. Let $\beta(x,i)<+\infty$. Then there is a  $S\subseteq V_x(G)$ of size $\beta(x,i)$  such that 
%\begin{itemize}
%\item[(i)] $v_i\in S$,
%\item[(ii)] $S$ contains at most $d$ $z$-vertices for each $z\in V(T_x)$,
%\item[(iii)] for each $z\in V(T_x)$, $C_z=c(S\cap V_G(z))$, and
%\item[(iv)]$S$ dominates all the vertices of $V_x(G)$.
%\end{itemize}
We consider the partition $(S_x,S_{z_1},\ldots,S_{z_s},S_{xz_1},\ldots,S_{xz_s})$ of $S$, where $S_{z_j}\subseteq V_{z_j}(G)$ and $S_{xz_j}$ are $e_j$-vertices for each $j\in\{1,\ldots,s\}$; some sets in the partition could be empty. 

Let $t=|S_x|$. For each $j\in \{1,\ldots,s\}$, select $v_{i(j)}\in S_x$ to be a maximum element of $S_x$ with respect to the relation $\preceq_{e_j}$. Observe that the selection 
is not necessarily unique. We consider a partition 
$\mathcal{P}=\{J_1,\ldots, J_t\}$ of 
$\{1,\ldots,s\}$ such that  
%$\{0,1,\ldots,s\}$ such that $0\in J_1$ and it holds that 
 $j,j'\in\{1,\ldots,s\}$ are in the same set of $\mathcal{P}$ if and only
if $i(j)=i(j')$. We define $\varphi\colon \{1,\ldots,t\}\rightarrow C_x$ as follows. If $J_h\neq\emptyset$, then $\varphi(h)=c(v_{i(j)})$ for $j\in J_h$. Then we extend $\varphi$ on $h\in\{1,\ldots,t\}$ with $J_t=\emptyset$ greedily to ensure that $\varphi$ is a surjection. Such a mapping always exists because $C_x=c(S_x\cap V_G(x))$.

Consider $z_j$ for $j\in\{1,\ldots,s\}$. It is assumed again that $W_j=\{u_1,\ldots,u_q\}$ is the set of $z_j$-vertices of $G$. We also assume that $u_1\preceq_{e_j}\ldots\preceq_{e_j}u_q$. 
Let $j\in J_{i'}$. We consider two cases depending on whether $z_j$ is loaded or not.

\paragraph{Case~1.} The vertex $z_j$ is loaded. Let $h\in\{1,\ldots,q\}$ be such that $u_h$ is a minimum element of $S_{z_j}$ with respect to $\preceq_{e_j}$.
By our inductive assumption, we have that $\beta(z_j,h)\leq |S_{z_j}|$, because $S_{z_j}$ is extendable. Consider 
\[
  X_{i(j),h}=\{w\in V(G)\mid w\in V_G(e_j),~M_w\cap(M_{v_{i(j)}}\cup M_{u_h})=\emptyset\}.
\]
By the definition of $i(j)$ and $h$, the vertices of $X_{i(j),h}$ are not dominated by $S_x\cup S_{z_j}$. Therefore, they are dominated by $S_{xz_j}$. By the definition of $\alpha_{e_j}$, we have that $\alpha_{e_j}(X_{i(j),h})\leq |S_{xz_j}|$.
Then 
\begin{align}
  \delta_j(x,i(j)) & =\min_{1\leq h'\leq q}\{\beta(z_j,h')+\alpha_{e_j}(X_{i',h'})\}\nonumber \\
                   & \leq\beta(z_j,h)+\alpha_{e_j}(X_{i(j),h})\nonumber\\
                   &\leq |S_{z_j}|+|S_{xz_j}|.\label{eq:delta-a}
\end{align}

\paragraph{Case~2.} The vertex $z_j$ is unloaded. Let $h\in\{0,\ldots,q\}$ be the minimum integer such that $u_{h+1},\ldots,u_{q}$ are dominated by $S_{z_j}$. Clearly, if $h>0$, then $u_{h}$ is not dominated by $S_{z_j}$. By the inductive assumption, we have that $\beta(z_j,h)\leq |S_{z_j}|$ as $S_{z_j}$ is extendable. 
Consider the set $Y_{i(j),h}$. By the definition of this set, we obtain that  the vertices of $X_{i(j),h}$ are not dominated by $S_x\cup S_{z_j}$. Therefore, they are dominated by $S_{xz_j}$. By the definition of $\alpha_{e_j}$, we have that $\alpha_{e_j}(Y_{i(j),h})\leq |S_{xz_j}|$.
Then 
\begin{align}
  \delta_j(x,i(j)) & = \min_{0\leq h'\leq q}\{\gamma(z_j,h')+\alpha_{e_j}(Y_{i',h'})\}\nonumber\\
                   & \leq\gamma(z_j,h)+\alpha_{e_j}(Y_{i(j),h}) \nonumber \\
                   &\leq |S_{z_j}|+|S_{xz_j}|.\label{eq:delta-b}
\end{align}

\medskip
Now we combine (\ref{eq:delta-a}) and (\ref{eq:delta-b}) and conclude that for each $h'\in\{1,\ldots, t\}$,
\begin{equation}\label{eq:delta}
\delta_{J_{h'}}(x,i')\leq \sum_{j\in J_{h'}}(|S_{z_j}|+|S_{xz_j}|),
\end{equation}
where $i'=i(j)$ for $j\in J_{h'}$.

By (\ref{eq:beta}) and (\ref{eq:delta}), we obtain that 
$$
\beta(x,i)\leq |S_x|+\sum_{j=1}^s(|S_{z_j}|+|S_{xz_j}|)=|S|.
$$

\medskip
Our next aim is to prove \autoref{claimB} for $\gamma$. Assume that $x$ is unloaded and let $i\in\{0,\ldots,p\}$. 

We suppose that $\gamma(x,i)<+\infty$ and prove that there is  a set $S\subseteq V_x(G)$ with $|S|\leq\gamma(x,i)$ such that 
\begin{enumerate}[(i)]
  \setcounter{enumi}{4}
\item[(v)] $v_{i+1},\ldots,v_{p}$ are dominated by $S$,
\item[(vi)] $S$ contains at most $d$ $z$-vertices for each $z\in V(T_x)$, 
\item[(vii)] for each $z\in V(T_x)$, $C_z=c(S\cap V_G(z))$, and
\item[(viii)] $S$ dominates all the vertices of $V_x(G)\setminus \{v_1,\ldots,v_p\}$.
\end{enumerate}

Since $\gamma(x,i)<+\infty$, $\gamma(x,i)=\sum_{j=1}^s\eta(j)$. Recall that $\eta(j)$ is computed differently depending on whether $z_j$ is loaded or not. We consider 
each $j\in\{1,\ldots,s\}$ and analyze 
the corresponding cases. Recall that $v_{i_1^j},\ldots,v_{i_p^j}$ is the ordering of the vertices of $W$ with respect to $\preceq_{e_j}$.

\paragraph{Case~1.} The vertex $z_j$ is loaded. 

%Let $h\in\{1,\ldots, q\}$ be such that the minimum in the right part of (\ref{eq:eta-a}) achieved for this value, that is, 
 %$\eta(j)= \beta(z_j,h)+\alpha_{e_j}(X_h)$, where $X_{h}=\{w\in V(G)\mid  w\in V_G(e_j),~M_w\cap M_{u_h}=\emptyset\}$.
We select $t\in\{1,\ldots,p+1\} $ for which the minimum achieved in (\ref{eq:psi-a}), that is, $\psi(j)=t$. In particular, we have that
\begin{equation}\label{eq:eta-one}
\eta(j)=\min_{1\leq h\leq q} (\beta(z_j,h)+\alpha_{e_j}(X_{h,t})), 
\end{equation}
where 
\[
  X_{h,t}=
\begin{cases}
X_{h}\cup\{v_{i_t^j},\ldots,v_{i_p^j}\} &\mbox{ if }t\leq p,\\
X_{h}&\mbox{ if }t=p+1
\end{cases}
\]
and $X_{h}=\{w\in V(G)\mid  w\in V_G(e_j),~M_w\cap M_{u_h}=\emptyset\}$.
Let $h\in\{1,\ldots,q\}$ be such that the minimum in the right part of (\ref{eq:eta-one}) is achieved for this value. 

 By the inductive assumption, there is a set  $S_{z_j}\subseteq V_{z_j}(G)$ of size at most $\beta(z_j,h)$  such that 
\begin{enumerate}[(i')]
\item $u_h\in S_j$,
\item $S_{z_j}$ contains at most $d$ $z$-vertices for each $z\in V(T_{z_j})$,
\item for each $z\in V(T_{z_j})$, $C_z=c(S_{z_j}\cap V_G(z))$, and
\item $S_{z_j}$ dominates all the vertices of $V_{z_j}(G)$.
\end{enumerate}
By the definition of $\alpha_{e_j}(X_{h,t})$, there is a set of $e_j$-vertices $S_{xz_j}$ of size $\alpha_{e_j}(X_{h,t})$ that dominates $X_{h,t}$. 
Note that $|S_{z_j}|+|S_{xz_j}|\leq \eta(j)$. Observe also that $S_{z_j}\cup S_{xz_j}$ dominates all vertices of $V_{z_j}(G)$ and the $e_j$-vertices.

\paragraph{Case~2.} The vertex $z_j$ is unloaded. 

We select $t\in\{1,\ldots,p+1\} $ for which the minimum achieved in (\ref{eq:psi-b}), that is, $\psi(j)=t$. In particular, we have that it holds 
\begin{equation}\label{eq:eta-two}
\eta(j)=\min_{1\leq h\leq q} (\gamma(z_j,h)+\alpha_{e_j}(Y_{h,t})), 
\end{equation}
where 
$$Y_{h,t}=
\begin{cases}
Y_h\cup\{v_{i_t^j},\ldots,v_{i_p^j}\} &\mbox{ if }t\leq p,\\
Y_h&\mbox{ if }t=p+1,  
\end{cases}
$$
and 
$$
Y_h=
\begin{cases}
\{w\in V(G)\mid   w\in V_G(e_j)\}&\mbox{ if }h=0,\\
\{w\in V(G)\mid   w\in V_G(e_j)\}\cup\{u_1,\ldots,u_h\}&\mbox{ if }h\geq 1.
\end{cases}
$$
Let $h\in\{1,\ldots,q\}$ be such that the minimum in the right part of (\ref{eq:eta-two}) is achieved for this value. 
By the inductive assumption, there is a set  $S_{z_j}\subseteq V_{z_j}(G)$ of size at most $\gamma(z_j,h)$  such that 
\begin{enumerate}[(i')]
  \setcounter{enumi}{4}
\item $u_{h+1},\ldots,u_{q}$ are dominated by $S_{z_j}$,
\item $S_{z_j}$ contains at most $d$ $z$-vertices for each $z\in V(T_{z_j})$, 
\item for each $z\in V(T_{z_j})$, $C_z=c(S_{z_j}\cap V_G(z))$, and
\item $S_{z_j}$ dominates all the vertices of $V_x(G)\setminus \{u_1,\ldots,u_q\}$.
\end{enumerate}
By the definition of $\alpha_{e_j}(Y_{h,t})$, there is a set of $e_j$ vertices $S_{xz_j}$ of size $\alpha_{e_j}(Y_{h,t})$ that dominates $Y_{h,t}$. 
We have that  $|S_{z_j}|+|S_{xz_j}|\leq \eta(j)$. Notice that  $S_{z_j}\cup S_{xz_j}$ dominates all vertices of $V_{z_j}(G)$ and the $e_j$-vertices.

\medskip
Now we define 
$$S=\cup_{j=1}^s(S_{z_j}\cup S_{xz_j}).$$

We have that 
$$|S|=\sum_{j=1}^s(|S_{z_j}|+|S_{xz_j}|)\leq \sum_{j=1}^s\eta(j)=\gamma(x,i).$$
By the definition of $S$, we have that \eqref{gamma6}--\eqref{gamma8} are fulfilled. To show~\eqref{gamma5}, recall that $\gamma(x,i)<+\infty$. Then 
for each $h\in\{i+1,\ldots,p\}$, there is $j\in\{1,\ldots,s\}$ such that $v_{i_{\psi(j)}^j}\preceq_{e_j} v_h$ and, therefore, 
$v_{h}$ is dominated by $S_{xz_j}$.

\medskip
Assume now that $S\subseteq V_x(G)$ is a set of minimum size satisfying \eqref{gamma5}--\eqref{gamma8} for $x$ and $i\in\{1,\ldots,p\}$ and $S$ is extendable. Because $S$ is a set of minimum size satisfying these conditions, $|S|\leq\gamma(x,i)$. We show that $|S|\geq \gamma(x,i)$.

We consider the partition $(S_{z_1},\ldots,S_{z_s},S_{xz_1},\ldots,S_{xz_s})$ of $S$, where $S_{z_j}\subseteq V_{z_j}(G)$ and $S_{xz_j}$ are $e_j$-vertices for $j\in\{1,\ldots,s\}$; some sets in the partition could be empty. Note that since $S$ is extendable, all the sets in the partition are extendable as well.

%Suppose that $\gamma(x,i)<+\infty$. Then  $\gamma(x,i)=\sum_{j=1}^s\eta(j)$. Let $j\in\{1,\ldots,s\}$. 

First, we show that $\eta(j)\leq |S_{z_j}|+|S_{xz_j}|$ for $j\in \{1,\ldots,s\}$.

If $x_j$ is loaded then $\eta(j)$ is computed by (\ref{eq:eta-a}). Let $h\in\{1,\ldots,q\}$ be such that  $u_h$ is a minimum with respect to $\preceq_{e_j}$ vertex in $S_{z_j}$. By induction, $|S_{z_j}|\geq \beta(z_j,h)$. Since the vertices of $G$ that are in $X_{h}=\{w\in V(G)\mid  w\in V_G(e_j),~M_w\cap M_{u_h}=\emptyset\}$ are not dominated by $S_{z_j}$, they are dominated by $S_{xz_{j}}$. By the definition of $\alpha_{e_j}(X_h)$, 
$|S_{xz_j}|\geq\alpha_{e_j}(X_h)$. Therefore,  
$\eta(j)\leq \beta(z_j,h)+\alpha_{e_j}(X_h)\leq |S_{z_j}|+|S_{xz_j}|$. 

If $z_j$ is unloaded, then $\eta(j)$ is computed by (\ref{eq:eta-b}).  We find minimum $h\in\{0,\ldots,q\}$  such that  $u_{h+1},\ldots,u_q$ are dominated by  $S_{z_j}$. 
By induction, $|S_{z_j}|\geq \gamma(z_j,h)
$
We consider 
$$Y_h=
\begin{cases}
\{w\in V(G)\mid   w\in V_G(e_j)\}&\mbox{ if }h=0,\\
\{w\in V(G)\mid   w\in V_G(e_j)\}\cup\{u_1,\ldots,u_h\}&\mbox{ if }h\geq 1,
\end{cases}$$ 
and observe that the vertices of $Y_h$ are not dominated by $S_{z_j}$. Hence, the vertices of $Y_h$ are dominated by $S_{xz_j}$. By the definition of $\alpha_{z_j}(Y_h)$, $|S_{xz_j}|\geq\alpha(Y_h)$. Hence,
$\eta(j)\leq \gamma(z_j,h)+\alpha_{e_j}(X_h)\leq |S_{z_j}|+|S_{xz_j}|$. 

Since $\eta(j)\leq |S_{z_j}|+|S_{xz_j}|$ for all $j\in\{1,\ldots,s\}$, we have that if $\gamma(x,i)=\sum_{j=1}^s\eta(j)$, then 
$$\gamma(x,i)=\sum_{j=1}^s\eta(j)\leq \sum_{j=1}^s(|S_{z_j}|+|S_{xz_j}|)=|S|.$$

Suppose now that $\gamma(x,i)\neq \sum_{j=1}^s\eta(j)$. By (\ref{eq:gamma}), we have that there is 
$h'\in\{i+1,\ldots,p\}$ such that $v_{h'}\prec_{e_j}v_{i_{\psi(j)}^j}$ for all $j\in\{1,\ldots,s\}$. The vertex $v_{h'}$ is dominated by $S$ by the condition (v). Assume that $v_{h'}$ is dominated by $S_{z_j}\cup S_{xz_j}$ for $j\in \{1,\ldots,s\}$. Let $v_{h'}=v_{i_t^t}$ according to the ordering of $x$-vertices with respect to $\preceq_{e_j}$. We again consider two cases.

\paragraph{Case~1.} The vertex $z_j$ is loaded. 

By (\ref{eq:psi-a}), 
$\eta(j)<\min_{1\leq h\leq q}( \beta(z_j,h)+\alpha_{e_j}(X_{h,t}))$, where $X_{h}=\{w\in V(G)\mid  w\in V_G(e_j),~M_w\cap M_{u_h}=\emptyset\}$ and 
$$X_{h,t}=
\begin{cases}
X_h\cup\{v_{i_t^j},\ldots,v_{i_p^j}\} &\mbox{ if }t\leq p,\\
X_h&\mbox{ if }t=p+1.
\end{cases}
$$

Let $h\in\{1,\ldots,q\}$ be such that  $u_h$ is a minimum with respect to $\preceq_{e_j}$ vertex in $S_{z_j}$. By induction, $|S_{z_j}|\geq \beta(z_j,h)$. 
Since the $e_j$-vertices of $G$ that are  in $X_{h,t}$ are not dominated by $S_{z_j}$, they are dominated by $S_{xz_{j}}$. By the definition of $\alpha_{e_j}(X_{h,t})$,  $|S_{xz_j}|\geq\alpha_{e_j}(X_{h,t})$. This means that 
$\eta(j)<|S_{z_j}|+|S_{xz_j}|$.

Let $h^*\in\{1,\ldots,q\}$ be such that the minimum in the right part of (\ref{eq:eta-one}) is achieved for this value, that is, 
$\eta(j)=\beta(z_j,h)+\alpha_{e_j}(X_{h^*})$. By the inductive assumption, there is  a set  $S_{z_j}'\subseteq V_{z_j}(G)$ of size at most $\beta(z_j,h^*)$  such that 
\begin{enumerate}[(i')]
\item $u_{h^*}\in S_j'$,
\item $S_{z_j}'$ contains at most $d$ $z$-vertices for each $z\in V(T_{z_j})$,
\item for each $z\in V(T_{z_j})$, $C_z=c(S_{z_j}\cap V_G(z))$, and
\item $S_{z_j}'$ dominates all the vertices of $V_{z_j}(G)$.
\end{enumerate}
By the definition of $\alpha_{e_j}(X_{h^*})$, there is a set of $e_j$ vertices $S_{xz_j}'$ of size $\alpha_{e_j}(X_{h^*})$ that dominates $X_{h^*}$.
Note that $|S_{z_j}|+|S_{xz_j}|\leq \eta(j)$. Observe also that $S_{z_j}'\cup S_{xz_j}'$ dominates all vertices of $V_{z_j}(G)$ and the $e_j$-vertices.
Let $S'=S_{z_j}'\cup S_{xz_j}'\cup\{v_1\}$. This set dominates all the vertices of $V_x(G)$. Note that $|S'|\leq |S_{z_j}\cup S_{xz_j}|$. 

\paragraph{Case~2.} The vertex $z_j$ is unloaded. 

By (\ref{eq:psi-b}), 
$\eta(j)<\min_{1\leq h\leq q}( \gamma(z_j,h)+\alpha_{e_j}(Y_{h,t}))$, where
$$Y_{h,t}=
\begin{cases}
Y_h\cup\{v_{i_t^j},\ldots,v_{i_p^j}\} &\mbox{ if }t\leq p,\\
Y_h&\mbox{ if }t=p+1,  
\end{cases}
$$
and 
$$
Y_h=
\begin{cases}
\{w\in V(G)\mid   w\in V_G(e_j)\}&\mbox{ if }h=0,\\
\{w\in V(G)\mid   w\in V_G(e_j)\}\cup\{u_1,\ldots,u_h\}&\mbox{ if }h\geq 1.
\end{cases}
$$

Let $h\in\{0,\ldots,q\}$ be the minimum index such that $u_{h+1},\ldots,u_{q}$ are dominated by  $S_{z_j}$. By induction, $|S_{z_j}|\geq \gamma(z_j,h)$. 
Since the vertices of $G$ that are  in $Y_{h,t}$ are not dominated by $S_{z_j}$, they are dominated by $S_{xz_{j}}$. By the definition of $\alpha_{e_j}(X_{h,t})$,  $|S_{xz_j}|\geq\alpha_{e_j}(X_{h,t})$. It means that 
$\eta(j)<|S_{z_j}|+|S_{xz_j}|$.

Let $h^*\in\{1,\ldots,q\}$ be such that the minimum in the right part of (\ref{eq:eta-two}) is achieved for this value, that is, 
$\eta(j)=\gamma(z_j,h^*)+\alpha_{e_j}(Y_{h^*})$. By the inductive assumption, there is  a set  $S_{z_j}'\subseteq V_{z_j}(G)$ of size at most $\gamma(z_j,h^*)$  such that 
\begin{enumerate}[(i')]
  \setcounter{enumi}{4}
\item $u_{h^*+1},\ldots,u_{q}$ are dominated by $S_{z_j}$,
\item $S_{z_j}$ contains at most $d$ $z$-vertices for each $z\in V(T_{z_j})$, 
\item for each $z\in V(T_{z_j})$, $C_z=c(S_{z_j}\cap V_G(z))$, and
\item $S_{z_j}$ dominates all the vertices of $V_x(G)\setminus \{u_1,\ldots,u_q\}$.
\end{enumerate}
By the definition of $\alpha_{e_j}(X_{h})$, there is a set of $e_j$ vertices $S_{xz_j}'$ of size $\alpha_{e_j}(Y_{h^*})$ that dominates $Y_{h^*}$.
Note that $|S_{z_j}|+|S_{xz_j}|\leq \eta(j)$. Observe also that $S_{z_j}'\cup S_{xz_j}'$ dominates all vertices of $V_{z_j}(G)$ and the $e_j$-vertices.
Let $S'=S_{z_j}'\cup S_{xz_j}'\cup\{v_1\}$. This set dominates all the vertices of $V_x(G)$. Note that $|S'|\leq |S_{z_j}\cup S_{xz_j}|$. 

\medskip
Now we use the set $S'$ obtained in both cases to obtain a contradiction to the extendability of~$S$.
The set $S$ is extendable, that is, there is a dominating set $D$  of $G$ such that 
\begin{enumerate}[(a)]   
\item  $D$ has  at most $d$ $x$-vertices for $x\in V(T)$,
\item for each $x\in V(T)$, $C_x\subseteq c(D\cap V_G(x))$,
\end{enumerate} 
that has the minimum size and contains $C$ and the conditions of the promise is fulfilled:
 the number of nodes $z\in V(T)$ such that $D$ contains an $z$-vertex is maximum and  for each $z\in V(T)$, $C_z= c(D\cap V_G(z))$.
Let $D'=(D\setminus (S_{z_i}\cup S_{xz_j}))\cup S'$.  It is straightforward to see that $D$ is a dominating set and $|D'|\leq |D|$.
We also have that (a) and (b) are fulfilled for $D'$, because of (ii$'$), (iii$'$), (vi$'$) and (vii$'$), but this contradicts the condition that  
the number of nodes $x\in V(T)$ such that $D$ contains an $x$-vertex is maximum because $D'$ contains the $x$-vertex $v_1$.
Hence, $\gamma(x,i)\neq \sum_{j=1}^s\eta(j)$ and this completes the proof of the claim that $|S|\geq\gamma(x,i)$. We thus proved \autoref{claimB}.%
\cqed
\end{proof}

Now we are ready to complete the description of our algorithm for \textsc{Dominating Set Extension}. The algorithm computes the table of values of $\beta(r,i)$ if $r$ is loaded and the table of values of $\gamma(r,i)$ if $r$ is unloaded. If $r$ is loaded, we find the minimum value $\beta(r,i^*)$ in the table for $i$. By \autoref{claimB},
if $\beta(r,i^*)<+\infty$, then $G$ has a dominating set $S$ of size at most $\beta(r,i^*)$ containing at most  $d$ $x$-vertices for $x\in V(T)$ such that
 for each $x\in V(T)$,  $C_x=c(S\cap V_G(x))$. Moreover, if the promise is true, then $\beta(r,i^*)<+\infty$ and the minimum size of $S$ is $\beta(r,i^*)$.  If $r$ is unloaded, then we consider $\gamma(r,0)$ in the table 
for $r$. By \autoref{claimB}, if $\gamma(r,0)<+\infty$, then $G$ has a dominating set $S$ of size at most $\gamma(r,0)$ containing at most  $d$ $x$-vertices for $x\in V(T)$ such that
 for each $x\in V(T)$,  $C_x=c(S\cap V_G(x))$, and  if the promise is true, then $\gamma(r,0)<+\infty$ and the minimum size of $S$ is $\gamma(r,0)$.
It remains to check whether $\beta(r,i^*)\leq k$ or $\gamma(r,0)\leq k$ respectively  and return the answer. 

\medskip
To evaluate the running time, observe that to compute $\beta(x,i)$, we consider  all possible partitions $\mathcal{P}=\{J_1,\ldots, J_t\}$ for $1\leq t\leq d$ of $\{0,1,\ldots,s\}$ into non-empty sets such that $0\in J_1$ where $s$ is the number of children of $x$. Since $s\leq \ell$, we have that the number of partitions is $2^{\Oh(\ell\log d)}$. 
Then for each partition  $\mathcal{P}=\{J_1,\ldots, J_t\}$, we consider all possible surjections $\varphi\colon \{1,\ldots,t\}\rightarrow C_x$. Since $t\leq d$ and each $|C_x|\leq d$, there are $2^{\Oh(d\log d)}$  choices of $\varphi$. Because each value of $\alpha$ and $\gamma$ can be computed in polynomial time,
it implies that the total running time of the algorithms is $2^{\Oh((\ell+d)\log d)}n^{\Oh(1)}$.
This finishes the proof of~\autoref{lem:promise} about \textsc{Dominating Set Extension}.%
\qed
\end{proof}

Now we are ready to prove the main theorem of the section.

\begin{theorem}\label{thm:ds-leafage}
\textsc{Dominating Set} can be solved in time $2^{\Oh(\ell^2)}\cdot n^{\Oh(1)}$ for connected chordal graphs with leafage at most $\ell$.
\end{theorem}

\begin{proof}
Let $(G,k)$ be an instance of \textsc{Dominating Set} where $G$ is a connected chordal graph. 

We use the algorithm of Habib and Stacho~\cite{HabibS09} to compute its leafage $\ell(G)$. If $\ell(G)>\ell$, we stop and return a no-answer. 
 Otherwise, we consider the clique tree $T'$ of $G$ constructed by the algorithm. If $|T|=1$, then $G$ is a complete graph and \textsc{Dominating Set} has a straightforward solution. Let $\|T\|\geq 1$, that is, $\ell(G)\geq 2$.
We construct the tree $T$ from $T'$ by \emph{dissolving} nodes of degree two, that is, for a node $x$ of degree two with the neighbors $y$ and $z$, we delete $x$ and make $y$ and $z$ adjacent. 
Observe that since $T$ is a tree with at most $\ell$ leaves that has no node of degree two, $|T|\leq 2\ell-2$. We have that $G$ is a $T$-graph. Note also that the algorithm of Habib and Stacho~\cite{HabibS09} gives us a $T$-representation $\mathcal{M}=\{M_v\}_{v\in V(G)}$ where $M_v\in V(T')$ for $v\in V(G)$. 

We consider the $2^{|T|}-1\leq 2^{2\ell-2}-1$ non-empty subsets of $V(T)$ and construct a coloring $c\colon V_G(T)\rightarrow \{1,\ldots,2^{|T|}\}$ such that for $u,v\in V_G(T)$, $c(u)=c(v)$ if and only if $u$ and $v$ are $Q$-vertices for the same $Q\subseteq V(T)$. 

By \autoref{lem:upper-b}, a minimum dominating set of $G$ contains at most $2|T|-2\leq 4\ell-6$ vertices of $V_G(T)$. Clearly, these vertices can have at most $4\ell-6$ distinct colors. We consider all sets $C\subseteq\{1,\ldots,2^{|T|}\}$ of distinct colors of size at most $4\ell-6$ and for each $C$, we aim to find a minimum dominating set of $G$  whose vertices in $V_G(T)$  are colored by the maximum number of distinct colors and are colored exactly by the colors of $C$. Since we consider all possible choices of $C$, it holds for some $C$. 

Toward this aim, we apply the following rule.

\paragraph{Rule~1.} If there is an $xy$-vertex $w$ of $G$ for $xy\in E(T)$ and
\begin{enumerate}[(i)]
\item $x,y\notin M_u$ for every $u\in V_G(T)$ with $c(u)\in C$; and
\item there is $v\in V_G(T)$ (with $c(v)\notin C$) such that $x,y\in M_v$,
\end{enumerate}
 then discard the current choice of $C$.

\medskip
To see that the rule is safe, observe that if $D$ is minimum dominating set of $G$ whose vertices in $V_G(T)$  are colored exactly by the colors of $C$, then  $w$ is dominated by some $xy$-vertex $w'$. We have that $v\notin D$, because $c(v)\notin C$. Then it is straightforward to see that $D'=(D\setminus \{w'\})\cup \{v\}$ is a minimum dominating set of $G$ whose vertices in $V_G(T)$  are colored  by $|C|+1$ colors. 

\medskip
Now we are looking for a dominating set $D$ of minimum size such that $c(D\cap V_G(T))=C$.

We use the following rule.

\paragraph{Rule~2.} If there is a $Q$-vertex $u$ of $G$ for non-empty $Q\subseteq V(T)$ such that
\begin{enumerate}[(i)]
\item $c(u)\notin C$; and
\item there is $c\in C$ such that for every $v\in V_G(T)$ with $c(v)=c$, $v$ dominates $u$,
\end{enumerate}
then delete $u$.

\medskip
To see that the rule is safe, observe that $u$ cannot be included in a dominating set $D$ of minimum size such that $c(D\cap V_G(T))=C$ and $u$ is dominated by any set  $D$  such that $c(D\cap V_G(T))=C$.

\medskip
%Assume that $C$ is the considered set of colors. 
Let 
\begin{align*}
  A =\{xy\in E(T)\mid & x,y\in M_u\text{ for some }u\in V_G(T)\text{ such that } c(u)\in C\}\\
  A'=\{xy\in E(T)\mid & x,y\in M_u\text{ for some }u\in V_G(T)\text{ such that } c(u)\notin C\text{ and }\\
 &x,y\notin M_v\text{ for }v\in V_G(T)\text{ such that } c(v)\in C\}.  
\end{align*}

Observe that because of Rule~1, there are no $e$-vertices for $e\in A'$.
We contract the edges $e\in A\cup A'$. Denote by $\hat{T}$ the tree obtained from $T$ tree and let $\hat{T}'$ be the tree obtained from $T'$ by contracting the paths that correspond to the contracted edge. We also construct the graph $\hat{G}$ that is obtained from $G$ by contracting these edges of $T$ and we also construct its $\hat{T}$-representation $\hat{\mathcal{M}}=\{\hat{M}_v\}_{v\in V(\hat{G})}$ where $\hat{M}_v\in V(\hat{T}')$ for $v\in V(\hat{G})$. 
We set $\hat{c}=c|_{V(\hat{G})}$ and for every $x\in V(\hat{T})$ define  
\[
  C_x=\{c,\ \exists u\in  V_{\hat{G}}(x)\text{ s.t. }\hat{c}(u)=c\}.
\]
Observe that $\hat{\mathcal{M}}$ is a nice $\hat{T}$-representation of $\hat{G}$. Indeed, for every $xy\in E(T)$ such that $x,y\in M_u$ for $u\in V_G(T)$ we have that $xy\in A$ if $c(u)\in C$ and $xy\in A'$ if $c(u)\notin C$ because of Rule~2, and all such edges $xy$ are contracted.

Combining Lemmas~\ref{lem:contraction} and \ref{lem:contraction-uncolored} we obtain that $D$ is a dominating set of minimum size with 
 $C=c(D\cap V_G(T))$ if and only if $D$ is a dominating set of $G'$ of minimum size such that $C=\hat{c}(D\cap V_{\hat{G}}(\hat{T}))$. Note that the condition 
 $C=\hat{c}(D\cap V_{\hat{G}}(\hat{T}))$ is equivalent to the condition that for every $x\in V(\hat{T})$, $C_x=\hat{c}(D\cap V_{\hat{G}}(x))$, because the $\hat{T}$-representation of $\hat{G}$ is nice and $C_x\cap C_y=\emptyset$ for distinct $x,y\in V(\hat{T})$.
 
 We set $d=|T|+\ell-1\leq 3\ell-3$ and apply the next rule.
 \paragraph{Rule~3.} If there is $x\in V(\hat{T})$ with $|C_x|>d$, then discard the current choice of $C$.

\medskip
To see that the rule is safe, assume that the input graph $G$ has a minimum dominating set $D$ whose vertices in $V_G(T)$  are colored exactly by the colors of $C$. 
 By \autoref{lem:upper-a}, we have that if a set of nodes $X$ of $T$ is contracted into a single vertex $x$ of $\hat{T}$, then $D$ has at most $|X|+\ell-1$ vertices whose models contain a vertex of $X$ and, therefore, the number of vertices colored by the colors  of $C_x$ in $D$ is at most $d$.
 
\medskip
We arbitrarily select a node $r$ to be the root of $\hat{T}$ and $\hat{T}'$ respectively. 
Then we apply \autoref{lem:promise} to the instance $(\hat{T},k,d,\hat{c},\{C_x\}_{x\in V(\hat{T})})$ of  \textsc{Dominating Set Extension}.  
 
 Recall that  \textsc{Dominating Set Extension} is a promise problem. If the algorithm from \autoref{lem:promise} returns a yes-answer, it means that there is a dominating set $D$ of $\hat{G}$ of size at most $k$ such that for each $x\in V(\hat{T})$, $C_x=c(D\cap V_{\hat{G}}(x))$. This means that the input graph $G$ has a dominating set of size at most $k$. Still, if the promise is false, the algorithm can return an incorrect no-answer. 
 Recall that the promise of \textsc{Dominating Set Extension}  is the following: for every dominating set $D$ of $\hat{G}$ of minimum size with the properties that  
\begin{enumerate}[(a)]   
\item  $D$ has  at most $d$ $x$-vertices for $x\in V(\hat{T})$,
\item for each $x\in V(\hat{T})$, $C_x\subseteq c(D\cap V_G(\hat{x}))$,
\end{enumerate}
it holds that the number of nodes $x\in V(\hat{T})$ such that $D$ contains an $x$-vertex is maximum and  for each $x\in V(\hat{T})$, $C_x= c(D\cap V_G(\hat{x}))$.
By Lemmas~\ref{lem:contraction} and \ref{lem:contraction-uncolored}, we have that if $C$ is chosen in such a way that  $G$ has a minimum dominating set $D$ 
that has the maximum number of vertices of $V_G(T)$ and whose vertices in $V_G(T)$  are colored exactly by the colors of $C$, then this promise holds for the corresponding instance of \textsc{Dominating Set Extension} constructed for this choice of $C$. Therefore, if $(G,k)$ is a yes-instance of \textsc{Dominating Set}, then for some choice of $C$, we obtain a yes-answer. 

\medskip 
To evaluate the running time of the algorithm, observe that $T$, $T'$ and the representation  $\mathcal{M}$ are constructed in polynomial time by the algorithm of 
Habib and Stacho~\cite{HabibS09}. The coloring function $c\colon V_G(T)\rightarrow \{1,\ldots,2^{|T|}\}$ can be constructed in time $2^{\Oh(\ell)}\cdot n^{\Oh(1)}$.
Then we construct $2^{\Oh(\ell^2)}$ sets $C\subseteq\{1,\ldots,2^{|T|}\}$ of size at most $6\ell-8$ and it can be done in time $2^{\Oh(\ell^2)}$.
For each $C$, the Rules~1 and 2 can be applied in polynomial time. Similarly, the construction of the instance $(\hat{T},k,d,\hat{c},\{C_x\}_{x\in V(\hat{T})})$ of  \textsc{Dominating Set Extension} for a given $C$ can be done in polynomial time. Clearly, Rule~3 can be applied in polynomial time. Then we solve 
the constructed instance of \textsc{Dominating Set Extension}  in time $2^{\Oh(\ell\log\ell)}\cdot n^{\Oh(1)}$. Hence, the total running time of the algorithm is $2^{\Oh(\ell^2)}\cdot n^{\Oh(1)}$.%
\qed
\end{proof}

The theorem immediately gives the following corollary for $T$-graphs.

\begin{corollary}\label{ds:trees}
\textsc{Dominating Set} can be solved in time $2^{\Oh(|T|^2)}\cdot n^{\Oh(1)}$ for $T$-graphs if $T$ is a tree.
\end{corollary}

\subsection{A polynomial kernel for Clique}\label{sec:clique}
It was observed in \cite{ChaplickZ17} that the \textsc{Clique} problem  is \classFPT for $H$-graphs when parameterized by the solution size $k$ and $\|H\|$ (even when no $H$-representation of $G$ is given). We show that \textsc{Clique} admits a polynomial kernel when a representation is given.

Let $G$ be an $H$-graph with an $H$-representation $\mathcal{M}=\{M_v\}_{v\in V(G)}$ where, for the corresponding subdivision $H'$ of $H$ and $v \in V(G)$, $M_v\subseteq V(H')$.
Recall that for $e\in E(H)$, $v\in V(G)$ is an $e$-vertex if $M_v$ contains only subdivision nodes of $H'$ from the path in $H'$ corresponding to $e$ in $H$. 
We claim that we can find a maximum clique in $G$ that contains some $e$-vertex in polynomial time.

\begin{lemma}\label{lem:clique-poly}
Let $G$ be an $H$-graph given together with its $H$-representation. Then a clique of maximum size that contains at least one $e$-vertex for some $e\in E(G)$ can be found in time $\Oh(n^{3/2}m)$.
\end{lemma}

\begin{proof}
Let $\mathcal{M}=\{M_v\}_{v\in V(G)}$ be an $H$-representation of $G$.
For each $e$-vertex $u$ of $G$, we find a maximum clique $K$ such that $M_u$ is inclusion minimal for  $K$, that is, there is no $v\in K$ with $M_v\subset M_u$.
Let $e=xy$ for $x,y\in V(H)$ and denote by $P$ the $(x,y)$-path corresponding to $e$ in the subdivision $H'$ of $G$. Since $u$ is an $e$-vertex, $M_u\subseteq V(P)$, that is, the nodes of $M_u$ form a subpath of $P$. Denote by $x'$ and $y'$ the end-vertices of the subpath. Note that it can happen that $x'=y'$. Because $M_u$ is an inclusion minimal model of a vertex of $K$, for every $v\in K$, $x'\in M_v$ or $y'\in M_v$. Consider 
$U=\{v\in V(G)\mid x'\in M_v\text{ or }y'\in M_v\}$. We have that finding $K$ in $G$ is equivalent to finding a maximum clique containing $u$ in $G'=G[U]$.

Notice that $U$ can be partitioned into two cliques $K_1=\{v\in V(G)\mid x'\in M_v\}$ and $K_2=\{v\in V(G)\mid y'\in M_v\ \text{and}\ x'\notin M_v\}$. This means that $G'$ is a \emph{cobipartite} graph. A maximum clique in a cobipartite graph can be found in time $\Oh(\sqrt{n}m)$ by the algorithm of Hopcroft and Karp~\cite{HopcroftK73} as finding a maximum clique in $G'$ is equivalent to finding a maximum independent set in the complement of $G'$ that is a bipartite graph. Note that a maximum clique in $G'$ always contains $u$, because $u$ is adjacent to every other vertex of $G'$.

Since we consider all $e$-vertices to find a maximum clique containing some $e$-vertex for some $e\in E(H)$, the total running time is $\Oh(n^{3/2}m)$. \qed
\end{proof}

Now we a ready to construct our kernel.

\begin{theorem}\label{thm:kernel}
The \textsc{Clique} problem for $H$-graphs admits a kernel with at most $(k-1)|H|$ vertices if an $H$-representation of the input graph is given, which can be computed in $\Oh(n^{3/2}m)$-time.
\end{theorem}

\begin{proof}
Let $G$ be an $H$-graph with an $H$-representation $\mathcal{M}=\{M_v\}_{v\in V(G)}$ where $M_v\subseteq V(H')$ for the corresponding subdivision $H'$ of $H$.

First, we use \autoref{lem:clique-poly} to check whether $G$ has a clique of size at least $k$ that contains at least one $e$-vertex for some $e\in E(G)$. If we find such a clique we return a yes-answer. Assume that this is not the case. Let $G'$ be the graph obtained from $G$ by the deletion of all $e$-vertices for $e\in E(H)$. We have that $G$ has a clique of size at least $k$ if and only if $G'$ has a clique of size at least $k$.

If there is $x \in V(H)$ such that $V_x=\{v\in V(G')\mid x\in M_v\}$ has size at least $k$, then it is a clique of size at least $k$ and we return a yes-answer. Otherwise, we return $G'$. Clearly, $|G'|\leq (k-1)|H|$ in this case.%
\qed
\end{proof}

\section*{Acknowledgement}
We thank Saket Saurabh for bringing $H$-graphs to our attention.

%\addcontentsline{toc}{section}{\nameref{sec:ack}}

\bibliography{hgraphs} \bibliographystyle{plain}

\end{document}

%% file: Fig1.pdf_t
\begin{picture}(0,0)%
\includegraphics{Fig1.pdf}%
\end{picture}%
\setlength{\unitlength}{3947sp}%
\begingroup\makeatletter\ifx\SetFigFont\undefined%
\gdef\SetFigFont#1#2#3#4#5{%
  \reset@font\fontsize{#1}{#2pt}%
  \fontfamily{#3}\fontseries{#4}\fontshape{#5}%
  \selectfont}%
\fi\endgroup%
\begin{picture}(6911,3083)(1867,-1778)
\put(8264,-285){\makebox(0,0)[lb]{\smash{{\SetFigFont{12}{14.4}{\rmdefault}{\mddefault}{\updefault}{\color[rgb]{0,0,0}$x_2^{(j,i)}$}%
}}}}
\put(1882,-1385){\makebox(0,0)[lb]{\smash{{\SetFigFont{12}{14.4}{\rmdefault}{\mddefault}{\updefault}{\color[rgb]{0,0,0}$u_1$}%
}}}}
\put(3475,1122){\makebox(0,0)[lb]{\smash{{\SetFigFont{12}{14.4}{\rmdefault}{\mddefault}{\updefault}{\color[rgb]{0,0,0}$u_2$}%
}}}}
\put(4895,-1331){\makebox(0,0)[lb]{\smash{{\SetFigFont{12}{14.4}{\rmdefault}{\mddefault}{\updefault}{\color[rgb]{0,0,0}$u_3$}%
}}}}
\put(2842,-279){\makebox(0,0)[lb]{\smash{{\SetFigFont{12}{14.4}{\rmdefault}{\mddefault}{\updefault}{\color[rgb]{0,0,0}$w_{1,2}$}%
}}}}
\put(3389,-1432){\makebox(0,0)[lb]{\smash{{\SetFigFont{12}{14.4}{\rmdefault}{\mddefault}{\updefault}{\color[rgb]{0,0,0}$w_{1,3}$}%
}}}}
\put(4289,-79){\makebox(0,0)[lb]{\smash{{\SetFigFont{12}{14.4}{\rmdefault}{\mddefault}{\updefault}{\color[rgb]{0,0,0}$w_{2,3}$}%
}}}}
\put(4989,381){\makebox(0,0)[lb]{\smash{{\SetFigFont{12}{14.4}{\rmdefault}{\mddefault}{\updefault}{\color[rgb]{0,0,0}$u_i$}%
}}}}
\put(8715,401){\makebox(0,0)[lb]{\smash{{\SetFigFont{12}{14.4}{\rmdefault}{\mddefault}{\updefault}{\color[rgb]{0,0,0}$u_j$}%
}}}}
\put(6655,-112){\makebox(0,0)[lb]{\smash{{\SetFigFont{12}{14.4}{\rmdefault}{\mddefault}{\updefault}{\color[rgb]{0,0,0}$x_p^{(i,j)}$}%
}}}}
\put(5249,534){\makebox(0,0)[lb]{\smash{{\SetFigFont{12}{14.4}{\rmdefault}{\mddefault}{\updefault}{\color[rgb]{0,0,0}$y_1^{(i,j)}$}%
}}}}
\put(8448,541){\makebox(0,0)[lb]{\smash{{\SetFigFont{12}{14.4}{\rmdefault}{\mddefault}{\updefault}{\color[rgb]{0,0,0}$y_1^{(j,i)}$}%
}}}}
\put(3415,-1705){\makebox(0,0)[lb]{\smash{{\SetFigFont{12}{14.4}{\rmdefault}{\mddefault}{\updefault}{\color[rgb]{0,0,0}a)}%
}}}}
\put(6835,-1699){\makebox(0,0)[lb]{\smash{{\SetFigFont{12}{14.4}{\rmdefault}{\mddefault}{\updefault}{\color[rgb]{0,0,0}b)}%
}}}}
\put(6569,515){\makebox(0,0)[lb]{\smash{{\SetFigFont{12}{14.4}{\rmdefault}{\mddefault}{\updefault}{\color[rgb]{0,0,0}$y_p^{(i,j)}$}%
}}}}
\put(7017,-220){\makebox(0,0)[lb]{\smash{{\SetFigFont{12}{14.4}{\rmdefault}{\mddefault}{\updefault}{\color[rgb]{0,0,0}$x_p^{(j,i)}$}%
}}}}
\put(6789,382){\makebox(0,0)[lb]{\smash{{\SetFigFont{12}{14.4}{\rmdefault}{\mddefault}{\updefault}{\color[rgb]{0,0,0}$w_{i,j}$}%
}}}}
\put(7116,589){\makebox(0,0)[lb]{\smash{{\SetFigFont{12}{14.4}{\rmdefault}{\mddefault}{\updefault}{\color[rgb]{0,0,0}$y_p^{(j,i)}$}%
}}}}
\put(5415,-358){\makebox(0,0)[lb]{\smash{{\SetFigFont{12}{14.4}{\rmdefault}{\mddefault}{\updefault}{\color[rgb]{0,0,0}$x_2^{(i,j)}$}%
}}}}
\put(5069,-119){\makebox(0,0)[lb]{\smash{{\SetFigFont{12}{14.4}{\rmdefault}{\mddefault}{\updefault}{\color[rgb]{0,0,0}$x_1^{(i,j)}$}%
}}}}
\put(8608,-52){\makebox(0,0)[lb]{\smash{{\SetFigFont{12}{14.4}{\rmdefault}{\mddefault}{\updefault}{\color[rgb]{0,0,0}$x_1^{(j,i)}$}%
}}}}
\end{picture}%

%% file: Fig2.pdf_t
\begin{picture}(0,0)%
\includegraphics{Fig2.pdf}%
\end{picture}%
\setlength{\unitlength}{3947sp}%
\begingroup\makeatletter\ifx\SetFigFont\undefined%
\gdef\SetFigFont#1#2#3#4#5{%
  \reset@font\fontsize{#1}{#2pt}%
  \fontfamily{#3}\fontseries{#4}\fontshape{#5}%
  \selectfont}%
\fi\endgroup%
\begin{picture}(5050,1095)(997,-888)
\put(3202, 24){\makebox(0,0)[lb]{\smash{{\SetFigFont{12}{14.4}{\rmdefault}{\mddefault}{\updefault}{\color[rgb]{0,0,0}$M_{r_{s,t}^{(i,j)}}$}%
}}}}
\put(3329,-662){\makebox(0,0)[lb]{\smash{{\SetFigFont{12}{14.4}{\rmdefault}{\mddefault}{\updefault}{\color[rgb]{0,0,0}$w_{i,j}$}%
}}}}
\put(5529,-748){\makebox(0,0)[lb]{\smash{{\SetFigFont{12}{14.4}{\rmdefault}{\mddefault}{\updefault}{\color[rgb]{0,0,0}$u_j$}%
}}}}
\put(1288,-809){\makebox(0,0)[lb]{\smash{{\SetFigFont{12}{14.4}{\rmdefault}{\mddefault}{\updefault}{\color[rgb]{0,0,0}$u_i$}%
}}}}
\put(4314,-210){\makebox(0,0)[lb]{\smash{{\SetFigFont{12}{14.4}{\rmdefault}{\mddefault}{\updefault}{\color[rgb]{0,0,0}$y_{p-t}^{(j,i)}$}%
}}}}
\put(3928,-391){\makebox(0,0)[lb]{\smash{{\SetFigFont{12}{14.4}{\rmdefault}{\mddefault}{\updefault}{\color[rgb]{0,0,0}$y_{p-t+1}^{(j,i)}$}%
}}}}
\put(1862,-543){\makebox(0,0)[lb]{\smash{{\SetFigFont{12}{14.4}{\rmdefault}{\mddefault}{\updefault}{\color[rgb]{0,0,0}$x_{s-1}^{(i,j)}$}%
}}}}
\put(2241,-683){\makebox(0,0)[lb]{\smash{{\SetFigFont{12}{14.4}{\rmdefault}{\mddefault}{\updefault}{\color[rgb]{0,0,0}$x_{s}^{(i,j)}$}%
}}}}
\put(2253,-223){\makebox(0,0)[lb]{\smash{{\SetFigFont{12}{14.4}{\rmdefault}{\mddefault}{\updefault}{\color[rgb]{0,0,0}$y_{p-s}^{(i,j)}$}%
}}}}
\put(2707,-369){\makebox(0,0)[lb]{\smash{{\SetFigFont{12}{14.4}{\rmdefault}{\mddefault}{\updefault}{\color[rgb]{0,0,0}$y_{p-s+1}^{(i,j)}$}%
}}}}
\put(4441,-630){\makebox(0,0)[lb]{\smash{{\SetFigFont{12}{14.4}{\rmdefault}{\mddefault}{\updefault}{\color[rgb]{0,0,0}$x_{t}^{(j,i)}$}%
}}}}
\put(4883,-509){\makebox(0,0)[lb]{\smash{{\SetFigFont{12}{14.4}{\rmdefault}{\mddefault}{\updefault}{\color[rgb]{0,0,0}$x_{t-1}^{(j,i)}$}%
}}}}
\put(1075, 11){\makebox(0,0)[lb]{\smash{{\SetFigFont{12}{14.4}{\rmdefault}{\mddefault}{\updefault}{\color[rgb]{0,0,0}$M_{z_s^i}$}%
}}}}
\put(5681,-36){\makebox(0,0)[lb]{\smash{{\SetFigFont{12}{14.4}{\rmdefault}{\mddefault}{\updefault}{\color[rgb]{0,0,0}$M_{z_t^j}$}%
}}}}
\end{picture}%

%% file: Fig3.pdf_t
\begin{picture}(0,0)%
\includegraphics{Fig3.pdf}%
\end{picture}%
\setlength{\unitlength}{3947sp}%
\begingroup\makeatletter\ifx\SetFigFont\undefined%
\gdef\SetFigFont#1#2#3#4#5{%
  \reset@font\fontsize{#1}{#2pt}%
  \fontfamily{#3}\fontseries{#4}\fontshape{#5}%
  \selectfont}%
\fi\endgroup%
\begin{picture}(2951,3168)(1283,-2572)
\put(4219,-120){\makebox(0,0)[lb]{\smash{{\SetFigFont{12}{14.4}{\rmdefault}{\mddefault}{\updefault}{\color[rgb]{0,0,0}$v_5$}%
}}}}
\put(2919,-1046){\makebox(0,0)[lb]{\smash{{\SetFigFont{12}{14.4}{\rmdefault}{\mddefault}{\updefault}{\color[rgb]{0,0,0}$x_2$}%
}}}}
\put(3252, -6){\makebox(0,0)[lb]{\smash{{\SetFigFont{12}{14.4}{\rmdefault}{\mddefault}{\updefault}{\color[rgb]{0,0,0}$x_3$}%
}}}}
\put(2666,-2213){\makebox(0,0)[lb]{\smash{{\SetFigFont{12}{14.4}{\rmdefault}{\mddefault}{\updefault}{\color[rgb]{0,0,0}$x_1$}%
}}}}
\put(2919,-1773){\makebox(0,0)[lb]{\smash{{\SetFigFont{12}{14.4}{\rmdefault}{\mddefault}{\updefault}{\color[rgb]{0,0,0}$M_{z_1}$}%
}}}}
\put(2199,-1339){\makebox(0,0)[lb]{\smash{{\SetFigFont{12}{14.4}{\rmdefault}{\mddefault}{\updefault}{\color[rgb]{0,0,0}$M_{z_2}$}%
}}}}
\put(3446,-666){\makebox(0,0)[lb]{\smash{{\SetFigFont{12}{14.4}{\rmdefault}{\mddefault}{\updefault}{\color[rgb]{0,0,0}$M_{z_3}$}%
}}}}
\put(2586,-373){\makebox(0,0)[lb]{\smash{{\SetFigFont{12}{14.4}{\rmdefault}{\mddefault}{\updefault}{\color[rgb]{0,0,0}$M_{z_4}$}%
}}}}
\put(2359,-2487){\makebox(0,0)[lb]{\smash{{\SetFigFont{12}{14.4}{\rmdefault}{\mddefault}{\updefault}{\color[rgb]{0,0,0}$v_2$}%
}}}}
\put(3326,-2494){\makebox(0,0)[lb]{\smash{{\SetFigFont{12}{14.4}{\rmdefault}{\mddefault}{\updefault}{\color[rgb]{0,0,0}$v_1$}%
}}}}
\put(1653,-327){\makebox(0,0)[lb]{\smash{{\SetFigFont{12}{14.4}{\rmdefault}{\mddefault}{\updefault}{\color[rgb]{0,0,0}$v_3$}%
}}}}
\put(3773,413){\makebox(0,0)[lb]{\smash{{\SetFigFont{12}{14.4}{\rmdefault}{\mddefault}{\updefault}{\color[rgb]{0,0,0}$v_4$}%
}}}}
\end{picture}%

%% file: Fig4.pdf_t
\begin{picture}(0,0)%
\includegraphics{Fig4.pdf}%
\end{picture}%
\setlength{\unitlength}{3947sp}%
\begingroup\makeatletter\ifx\SetFigFont\undefined%
\gdef\SetFigFont#1#2#3#4#5{%
  \reset@font\fontsize{#1}{#2pt}%
  \fontfamily{#3}\fontseries{#4}\fontshape{#5}%
  \selectfont}%
\fi\endgroup%
\begin{picture}(2836,2834)(1283,-2378)
\put(2799,-1755){\makebox(0,0)[lb]{\smash{{\SetFigFont{12}{14.4}{\rmdefault}{\mddefault}{\updefault}{\color[rgb]{0,0,0}$y$}%
}}}}
\put(2266,-1386){\makebox(0,0)[lb]{\smash{{\SetFigFont{12}{14.4}{\rmdefault}{\mddefault}{\updefault}{\color[rgb]{0,0,0}$M_{u}$}%
}}}}
\put(2632,-626){\makebox(0,0)[lb]{\smash{{\SetFigFont{12}{14.4}{\rmdefault}{\mddefault}{\updefault}{\color[rgb]{0,0,0}$x$}%
}}}}
\end{picture}%